\documentclass[conference]{IEEEtran}
\usepackage{textcomp}
\usepackage{xcolor}
\usepackage{amsfonts}
\usepackage{times}
\usepackage{soul}
\usepackage{url}
\usepackage[utf8]{inputenc}
\usepackage[small]{caption}
\usepackage{graphicx}
\usepackage{amsmath}
\usepackage{booktabs}
\usepackage{algorithm}
\usepackage{algpseudocode}
\usepackage{esvect}
\usepackage{booktabs}
\usepackage{array}
\usepackage{multirow}
\usepackage[export]{adjustbox}
\usepackage[capitalize]{cleveref}
\title{~\vspace{-20pt}\\An Adaptive and Fast Convergent Approach to \\ Differentially Private
Deep Learning\vspace{-10pt}}

\author{Zhiying Xu$^{1}$, Shuyu Shi$^1$, Alex X. Liu$^2$, Jun Zhao$^3$, Lin Chen$^{4}$
\vspace{.8mm}
\\
\fontsize{9}{9}\selectfont\itshape
$^1$University of Nanjing, China;~~~$^2$Michigan State University, USA\\
$^3$Nanyang Technological University, Singapore;~~~$^4$Yale University, USA \vspace{.8mm}
\\
\fontsize{9}{9}\selectfont\ttfamily\upshape
$^1$\{zyxu@smail. ssy@\}nju.edu.cn, $^2$alexliu@cse.msu.edu, $^3$junzhao@ntu.edu.sg, $^4$abratdarcy@gmail.com\vspace{-15pt}
}

\usepackage{pifont}
\usepackage{color}

\usepackage{courier}

\newcommand{\lc}[1]{}
\newcommand{\xu}[1]{{#1}}
\newcommand{\Var}{\mathrm{Var}}
\usepackage{url}  
\usepackage[utf8]{inputenc}
\usepackage[english]{babel}
\usepackage{amsthm}
\newtheorem{theorem}{Theorem}
\newtheorem{lemma}{Lemma}
\usepackage{tabularx}
\usepackage{amsmath}
\usepackage{subcaption}
\usepackage{mathrsfs}
\usepackage{url}
\usepackage{amssymb}
\usepackage{mathtools}
\usepackage{amsthm}
\usepackage{xspace}
\usepackage{breqn}

\newcommand{\ie}{\emph{i.e.}}
\newcommand{\eg}{\emph{e.g.}}
\newcommand{\Alg}{\textsc{AdaDp}\xspace}
\newcommand{\AlgDpSgd}{\textsc{DpSgd}\xspace}
\newcommand{\expect}{\mathbb{E}}
\renewcommand{\baselinestretch}{0.935}
\theoremstyle{definition}
\newtheorem{definition}{Definition}
\usepackage{fancyhdr}

\DeclarePairedDelimiter\ceil{\lceil}{\rceil}
\DeclarePairedDelimiter\floor{\lfloor}{\rfloor}

\def\BibTeX{{\rm B\kern-.05em{\sc i\kern-.025em b}\kern-.08em
		T\kern-.1667em\lower.7ex\hbox{E}\kern-.125emX}}
 
\begin{document}

\maketitle
\pagestyle{fancy}
\thispagestyle{fancy} \lhead{This full paper appears in the Proceedings of IEEE International Conference on Computer Communications (INFOCOM), held in April 2020.}  
\cfoot{\thepage}
\renewcommand{\headrulewidth}{0.4pt}
\renewcommand{\footrulewidth}{0pt}



\begin{abstract}
With the advent of the era of big data, deep learning has become a prevalent building block in a variety of machine learning or data mining tasks, such as 
\xu{signal processing, network modeling and traffic analysis}, to name a few. 
The massive user data crowdsourced plays a crucial role in the success of deep learning models.
However, 
it has been shown 
that user data may be inferred from trained neural models and thereby exposed to potential adversaries, which raises 
information security and privacy concerns.
To address this issue, recent studies leverage the technique of differential privacy to design
private-preserving deep learning algorithms.
Albeit successful at privacy protection,  differential privacy degrades the performance of neural models.
In this paper, we develop \Alg, an adaptive and fast convergent learning algorithm with a provable privacy guarantee.
\Alg significantly reduces the privacy cost by improving the convergence speed with an adaptive learning rate and mitigates the negative effect of  differential privacy upon the model accuracy by introducing adaptive noise. 
The performance of \Alg is evaluated on real-world datasets. Experiment results show that it outperforms state-of-the-art differentially private approaches in terms of both privacy cost and model accuracy.

\end{abstract}


\section{Introduction}
\subsection{Background and Motivation}
The past decade has witnessed the remarkable success of deep learning techniques in various machine learning\,\slash\,data mining tasks, such as 
signal processing \cite{wang2016csi}, network modeling \cite{wang2017spatiotemporal} and traffic analysis \cite{zhou2018deep}.
The great success relies heavily on the massive collection of user data, which, however, often raise severe privacy and security issues. 
For example, Fredrikson et al. \cite{fredrikson2014privacy}, demonstrates that the individual privacy information
in
the training dataset can be recovered by repeatedly querying the output
probabilities of a disease recognition classifier built upon a convolutional
neural
network (CNN).
Existing privacy concerns are likely to discourage users from sharing their data and thereby obstruct the future development of deep learning itself.

This paper studies the problem of user privacy protection in the training process of neural models. We consider the white-box scenario where an adversary has access to the parameters of a trained model. In this scenario, many service providers allow users to download models to their personal devices (e.g., computers and smart phones), and malicious users could analyze the parameters of the model which may expose personal information in the training dataset. 
%
\subsection{Limitations of Prior Art}
To address the privacy issue, several \emph{differential privacy} (DP)~\cite{dwork2006calibrating} based approaches were proposed, which may be classified into two categories: data obfuscation and gradient obfuscation. Data obfuscation based approaches obfuscate data with noise prior to potential exposure of sensitive information \cite{bost2015machine, zhang2018privacy}. These approaches may suffer from significant accuracy degradation of the trained model. 
The reason is that to guarantee the differential privacy bound, the added noise may be excessively intense and make differently labeled training instances almost indistinguishable.
In contrast to data obfuscation, gradient obfuscation based approaches add noise to the gradient in the training process \cite{kairouz2015secure, shokri2015privacy, abadi2016deep, lee2018concentrated, koskela2018learning, xiang2019differentially}. However, they may not circumvent the accuracy degradation issue completely. 
%
\xu{Although some methods aim to improve the accuracy of gradient obfuscation \cite{lee2018concentrated, xiang2019differentially}, they have three key limitations. First, the privacy cost is high because the convergence speed of these methods is slow while the privacy cost is accumulated for each gradient calculation. Second, the accuracy still cannot meet the high-precision requirements of many applications since they add identically distributed noise to all components of the gradient which results in large distortion of the original gradient.
Third, these methods are computationally inefficient because they need to evaluate the model multiple times or solve a large-scale optimization problem  per iteration which make the task computationally prohibitive.}

\subsection{Proposed Approach}
In this paper, we propose \Alg, an adaptive and fast convergent approach
to differentially private deep learning.
Our key observation is that different components of the gradient have inhomogeneous sensitivity to the training data. In light of this observation, \Alg mitigates the influence of noise on the model performance by adaptively sampling noise from different Gaussian distributions, based on the sensitivity of each component. In the first stage, \Alg adjusts the learning rate in an adaptive manner based on historical gradients such that infrequently updated components tend to have a larger learning rate. Then, \Alg samples
noise from different Gaussian distributions according to the
sensitivity of each gradient component and constructs differentially private gradients by adding the sensitivity-dependent noise to the  original gradient. In this way, Gaussian noise with a lower variance is added to components with a smaller sensitivity.

Compared to existing data obfuscation and gradient obfuscation based methods, \Alg has three key advantages. 
 
First, \xu{the privacy cost of \Alg is low because it} exhibits remarkable improvement in the convergence speed due to an adaptive learning rate. Since the privacy cost is accumulated on each gradient update, a faster rate of convergence indicates that training a model using \Alg incurs lower privacy cost.

Second, \Alg achieves both a provable privacy guarantee and a comparable accuracy to non-differentially private models simultaneously. We will show later that this is attained by adding adaptive noise to different gradient components, depending on their sensitivity. As the model converges, we will see a decrease in the expected sensitivity of each gradient component, which thereby reduces the variance of noise distribution. In other words, in contrast to prior works, the noise distribution of \Alg is adaptive to not only different gradient components but also different training iterations. 

Third, \Alg is computationally efficient since it does not need to solve any optimization problem to determine the noise distribution at each iteration, in sharp contrast to \cite{xiang2019differentially} that requires solving a large-scale non-convex optimization problems. We design an efficient scheme for adjusting the noise distribution and the scheme is evaluated via both theoretical analysis and numerical experiments.

\subsection{Technical Challenges and Solutions}

First, it is technically challenging to mathematically analyze the influence of the noise distribution upon the prediction performance, in light of the complicated nature of deep neural networks.
 As an alternative, we analyze the sufficient and necessary condition for noise distributions to guarantee the target differential privacy level.
 The condition involves an inequality with respect to the sensitivity of each gradient component and the variance of each corresponding Gaussian distribution. 
 Based on the analysis, we conduct an experiment to compare the influence of different noise distributions upon the original \xu{function which can be seen as a query on a dataset}.\lc{what's the original query function?}
 According to the theoretical and experiment results, we propose a heuristic that adapts the noise distributions to the sensitivity of each gradient component. Finally, we perform another experiment to verify the effectiveness of this heuristic according to the influence of noise on the gradient descent algorithm \xu{which is widely used for optimizing deep learning models}. \lc{why gradient descent here? without momentum?}

Another technical challenge is to compute the privacy cost without any assumptions on the parameters such as the noise level and the sampling ratio.
This is in sharp contrast to prior methods.
For example, the moments accountant method requires the noise level $\sigma \geq 1$ and the sampling ratio $q < \frac{1}{16\sigma}$ \cite{abadi2016deep}. 
To remove the assumptions,
we use a technique termed subsampled \emph{R\'enyi differential privacy} (RDP) \cite{mironov2017renyi}, which computes the privacy cost by analyzing the privacy amplification in the subsampling scenario. And importantly, it requires no assumption on the parameters in the analysis \cite{wang2018subsampled}.\lc{I skip this paragraph}

\subsection{Summary of Experiment Results}
We evaluated the privacy cost, the accuracy and the computational efficiency of \Alg and baselines on two real datasets: MNIST \cite{lecun1998gradient} and CIFAR-10 \cite{krizhevsky2009learning}. The results show that \xu{\Alg outperforms state-of-the-art methods with 50\% privacy cost reduction}. On MNIST and CIFAR-10, \Alg achieves an improvement of up to 4.3\% and 3.5\%, respectively, in accuracy over state-of-the-art methods. \xu{Our experimental results also show that the processing time of \Alg at each iteration is much less than that of \cite{xiang2019differentially, lee2018concentrated}, validating the computational efficiency of \Alg}.

\section{Related Work}
To protect sensitive information in crowdsourced data collection, 
differentially private crowdsourcing mechanisms were designed~\cite{zhang2018crowdbuy, jin2019if, niu2019making, li2019pedss}.
To preclude personal information from being inferred and/or identified from neural models~\cite{fredrikson2014privacy}, a line of works emerged which applied DP to deep learning~\cite{kairouz2015secure, shokri2015privacy, phan2016differential,abadi2016deep, papadimitriou2017dstress,papernot2016semi, phan2017adaptive, papernot2018scalable, lee2018concentrated,  koskela2018learning, collet2018boosting, wang2018geographic, xiang2019differentially}.
For instance, \cite{papernot2016semi, papernot2018scalable} integrate DP into a teacher-student framework to protect data on student nodes. \cite{shokri2015privacy, kairouz2015secure, mao2018privacy, collet2018boosting} study transferring features or gradients with privacy control in collaborative deep learning. \cite{zhang2018crowdbuy, jin2019if, niu2019making, li2019pedss} apply DP to personal data collection. In the above works, black-box attacks are (implicitly) assumed, \ie, the learned model is inaccessible to adversaries. Other works studied privacy protection in a more realistic white-box model where adversaries may have full knowledge of the  model~\cite{abadi2016deep,phan2017adaptive,lee2018concentrated,koskela2018learning,xiang2019differentially}.
\cite{abadi2016deep} proposes a differentially private gradient descent
algorithm
\AlgDpSgd by adding Gaussian noise to the gradient.
\cite{koskela2018learning} uses an adaptive learning rate to improve the convergence rate and reduce the privacy cost.
\cite{phan2017adaptive} introduces the Laplace mechanism such that the privacy budget consumption is independent of the number of training steps. \cite{lee2018concentrated} allocates different privacy
budgets to each training iteration to counteract the influence of noise on the gradient.  In all aforementioned works, however, the noise on each gradient component follows the \emph{same}
probability distribution. As a result, the original gradient is distorted to a
large extent. Although \cite{xiang2019differentially} samples noise from different distributions for each gradient component, solving a large-scale optimization is required at every step.

\section{Our Approach} \label{sec:approach}

To reduce the privacy cost, \Alg uses
 an adaptive learning rate for acceleration of the convergence.
 Additionally, \Alg adds inhomogeneous and adaptive noise to different coordinates of the gradient based on their  sensitivity in order to mitigate the influence of noise on the model performance.
The next two subsections elaborate the adaptive learning rate and noise respectively. Then we present \Alg and show its differential privacy guarantee.

\subsection{Adaptive Learning Rate}\label{sub:adaptive_learning_rate}
The most popular method for training deep models is the gradient-descent-type algorithms.
They iteratively update the parameters of a model
by moving them in the direction opposite to the gradient of the loss function
evaluated on the training data. The loss function
 $\mathscr{\mathcal{L}}$ is the difference
between the predictions and the true labels. To minimize the loss $\mathscr{\mathcal{L}}(\theta)$,
stochastic gradient descent (\textsc{Sgd}) randomly chooses a subset
of training data (denoted by $S$) at each iteration and performs the update
 $\theta\gets\theta-\eta\frac{1}{|S|}\sum_{i\in S}\nabla_{\theta}\mathcal{L}(\theta,\,x_{i})$.
\textsc{Sgd} poses several
challenges, \eg, selection of a proper learning rate
and avoidance of local minimum traps.

More advanced optimizers, including \textsc{RMSProp}, \textsc{Adam}, \textsc{Adadelta} and
\textsc{Nadam}, are proposed to address the above issues. They
adjust
the learning rate on a per-parameter basis in an adaptive manner and scale the coordinates
of the gradient according to historical data.

\Alg uses an adaptive strategy similar to that of \textsc{RMSProp}. Nevertheless, we would like to note that the framework proposed in this paper is applicable to other adaptive gradient-descent-type algorithms.
Recall the update of \textsc{RMSProp}
\begin{align}
	\label{eq:rmsprop}
	\expect[g^2]_t &\gets (1-\gamma)\expect[g^2]_{t-1} + \gamma (g_t)^2 \\ \nonumber
	\theta_{t} &\gets \theta_{t-1} - \eta\frac{g_t}{\sqrt{\expect[g^2]_{t} + \epsilon_0}},
\end{align}
where ${\theta_{t}}$ denotes the parameters at step ${t}$, ${g_t}$ denotes the original gradient,  ${\eta}$ is the learning rate, and ${\epsilon_0}$ is the smoothing term (in case that the denominator is 0).

In \Alg, let $\Delta\theta_{t}$ denote the update term at step $t$ and we have $\Delta\theta_{t}= \frac{\tilde{g}_t}{\sqrt{\expect[\tilde{g}^2]_{t} + \epsilon_0}}$ where $\tilde{g}_t$ is the noisy gradient which has been added Gaussian noise. In other words, \Alg uses the denominator $\sqrt{\expect[\tilde{g}^2]_{t} + \epsilon_0}$ to adjust the learning rate $\eta$ in an adaptive manner.

\subsection{Adaptive Noise}\label{sub:adaptive_noise}
The intuition of adaptive noise is that different coordinates of the gradient exhibit inhomogeneous sensitivities due to their different values.
It significantly affects the direction of the  gradient if noise with higher intensity is added to coordinates with a smaller value, and \emph{vice versa}.
In light of this intuition,
\Alg clips the  gradient
and adds Gaussian noise with a smaller/larger variance to dimensions of the clipped gradient with a smaller/larger sensitivity. 

Formally, we define the $\ell^2$-sensitivity of a function $f$ as $\Delta_f = \max_{D, D^\prime} \| f(D) - f(D^\prime) \|_2$ where $D, D^\prime$ are two datasets which differ in only one record. Correspondingly, adaptive noise is sampled from Gaussian distributions with different variances based on the $\ell^2$-sensitivity of each dimension. In the meanwhile, adaptive noise must satisfy the constraint of differential privacy. To guarantee this, we first present the following lemma which provides the theoretical foundation of adaptive noise.

\begin{lemma}
	\label{fst_lemma}
	Suppose mechanism $M(D) = f(D) + Z$, where $D$ is the input dataset, $f(\cdot)$ is a $m$-dimension function that $\Delta_f \leq 1$ and noise $Z \sim \mathcal{N}(0, \sigma_*^2)$, is $(\epsilon, \delta)$-differentially private, then mechanism $M^\prime(D) = f^\prime(D) + Z^\prime$, where $f^\prime(\cdot)$ is also a $m$-dimension function, $Z^\prime = (z^\prime_1, \dots, z^\prime_m)^T$ and $\forall i\in [m],\, z^\prime_i \sim \mathcal{N}(0, \sigma_i^2)$, is also $(\epsilon, \delta)$-differentially private if $\sum_{i=1}^{m}\frac{s_i^2}{\sigma_i^2} \leq \frac{1}{\sigma_*^2}$ where $s_i$ is the $\ell^2$-sensitivity of the $i$\textsuperscript{th} dimension of $f^\prime(\cdot)$.
\end{lemma}

To prove \cref{fst_lemma}, we need an auxiliary result which involves the sufficient and necessary condition on the privacy loss variable to satisfy $(\epsilon, \delta)$-DP.
\begin{lemma} [Analytical differential privacy \cite{balle2018improving}]
	\label{ana_theorem}
	A mechanism $M$ is $(\epsilon, \delta)$-DP if and only if for each $D, D^\prime$, the following holds:
	\begin{align}
	\label{eq:ana_gaus}
		\Pr\left( l_{M, D, D^\prime} \geq \epsilon  \right) - e^\epsilon\Pr\left(l_{M, D^\prime, D} \leq -\epsilon \right) < \delta,
	\end{align}
	where $l_{M, D, D^\prime}$ is the privacy loss variable defined by $\ln{\frac{\Pr\left(M(D)=o\right)}{\Pr\left(M(D^\prime)=o\right)}}$.
\end{lemma}
We are now ready to present the proof of  \cref{fst_lemma}. 
\begin{proof}
	The first step is to show that $l_{M, D, D^\prime}$ is a Gaussian random variable. Let $(r_1, \dots, r_m) = o - f(D)$. We consider the worst case of $l_{M, D, D^\prime}$. In this case, we have $(s_1, \dots, s_m) = f(D) - f(D^\prime)$, which yields $(r_1+s_1, \dots, r_m + s_m) = o - f(D^\prime)$. As a result, the following equations hold
	\begin{align}
		\ln\frac{\Pr(M(D)=o)}{\Pr(M(D^\prime)=o)}
		&= \ln  \frac{\prod_{j}^{m}\left(\frac{1}{\sqrt{2\pi\sigma_j^2}}\exp\left(\frac{-r_j^2}{2\sigma_j^2}\right)\right)}{\prod_{i}^{m}\left(\frac{1}{\sqrt{2\pi\sigma_i^2}}\exp\left(\frac{-(r_i + s_i)^2}{2\sigma_j^2}\right)\right)}  \nonumber  \\
		&=\sum_{j}^{m}\left(  \frac{-r_j^2}{2\sigma_j^2}  \right) - \sum_{j}^{m}\left( \frac{-(r_j + s_j) ^ 2}{2\sigma_j^2}  \right) \nonumber \\
		&= \sum_{j}^{m}\left(\frac{s_j^2}{2\sigma_j^2}  \right) + \sum_{j}^{m}\left( \frac{r_js_j}{\sigma_j^2}  \right).
	\end{align}
	
	In light of $r_j \sim \mathcal{N}(0, \sigma_j^2)$, we obtain that  $l_{M, D, D^\prime} =  \sum_{j}^{m}\frac{s_j^2}{2\sigma_j^2}   + \sum_{j}^{m} \frac{r_js_j}{\sigma_j^2}$ also obeys a Gaussian distribution. Specifically, we have $l_{M, D, D^\prime} \sim \mathcal{N}(\sum_{j}^{m}\frac{s_j^2}{2\sigma_j^2}, \sum_{j}^{m}\frac{s_j^2}{\sigma_j^2})$. If we let $H$ denote $\sum_{j}^{m}\frac{s_j^2}{\sigma_j^2}$, it can be re-written as $l_{M, D, D^\prime} \sim\mathcal{N}(\frac{H}{2}, H)$.
	
	The second step is to show that $\Pr\left( l_{M, D, D^\prime} \geq \epsilon  \right) - e^\epsilon \Pr\left( l_{M, D^\prime, D} \leq -\epsilon \right)$ is monotonically increasing in $H$. Let us compute the first term
	\begin{align*}
	\Pr\left(l_{M, D, D^\prime} \geq \epsilon \right) = \frac{1}{\sqrt{2\pi}}\int_{-\infty}^{A(H)}e^{-x^2/2}dx  
	\end{align*}where $A(H) = \frac{\sqrt{H}}{2} - \frac{\epsilon}{\sqrt{H}}$.
	Similarly, the second term can be re-written as
	\begin{align*}
	\Pr\left(l_{M, D^\prime, D} \leq -\epsilon\right) = \frac{1}{\sqrt{2\pi}}\int_{-\infty}^{A^\prime(H)}e^{-x^2/2}dx
	\end{align*}where $A^\prime(H) = \frac{-\epsilon}{\sqrt{H}} - \frac{\sqrt{H}}{2}$.
	
	Then the derivative of $\Pr\left(l_{M, D, D^\prime} \geq \epsilon  \right) - e^\epsilon \Pr\left( l_{M, D^\prime, D} \leq -\epsilon \right)$ about $H$ is:
	\begin{align*}
	&\frac{d\Pr\left( l_{M, D, D^\prime} \geq \epsilon  \right) - e^\epsilon \Pr\left( l_{M, D^\prime, D} \leq -\epsilon \right) }{dH} \\ \nonumber
	&= \frac{1}{\sqrt{2\pi}} \Big(e^{\frac{-(A(H))^2 }{2}} \cdot \frac{dA(H)}{dH} - e^{\epsilon}e^{-\frac{(A^\prime(H))^2}{2}}\cdot\frac{dA^\prime(H)}{dH}\Big) \\ \nonumber
	&=\frac{1}{\sqrt{2\pi}} \Big(e^{\frac{-(A(H))^2 }{2}}\big(\frac{1}{4\sqrt{H}} + \frac{\epsilon}{2}H^{3/2}\big) \\ \nonumber
	& \quad \quad \quad - e^\epsilon e^{\frac{-(A^\prime(H))^2}{2}} \big(\frac{-1}{4\sqrt{H}} + \frac{\epsilon}{2}H^{3/2}\big)\Big) \\ \nonumber
	&=\frac{1}{\sqrt{2\pi}} \Big(\frac{1}{4\sqrt{H}} \big(e^{\frac{-(A(H))^2 }{2}} + e^\epsilon e^{\frac{-(A^\prime(H))^2}{2}}\big) \\ \nonumber
	& \quad \quad \quad + \frac{\epsilon}{2}H^{3/2}\big(e^{\frac{-(A(H))^2 }{2}} - e^{\epsilon - \frac{(A^\prime(H))^2}{2}}\big) \Big).
	\end{align*}
	Since $(A^\prime(H))^2 = (A(H))^2 + 2\epsilon$, then we have 
	\begin{align*}
	&\frac{d\Pr\left( l_{M, D, D^\prime} \geq \epsilon  \right) - e^\epsilon \Pr\left( l_{M, D^\prime, D} \leq -\epsilon \right) }{dH} \\\nonumber
	&= \frac{1}{\sqrt{2\pi}} \left(\frac{1}{4\sqrt{H}} (e^{-\frac{(A(H))^2 }{2}} + e^\epsilon e^{-\frac{(A^\prime(H))^2}{2}})\right) \\ \nonumber
	&\geq 0.
	\end{align*}
	Therefore, $\Pr\left( l_{M, D, D^\prime} \geq \epsilon  \right) - e^\epsilon \Pr\left( l_{M, D^\prime, D} \leq -\epsilon \right)$ is monotonically increasing with $H$.
	
	Then if a function $f(\cdot)$ that $\Delta_f \leq 1$ satisfies $(\epsilon, \delta)$-DP with $\sigma_*$, the privacy loss variable $l^*_{M, D, D^\prime}$ must satisfy \cref{eq:ana_gaus}. Note that $l^*_{M, D, D^\prime}$ is a Gaussian variable that $l^*_{M, D, D^\prime} \sim\mathcal{N}(\frac{H_*}{2}, H_*)$ where $H_* = \sum_{j}^{m}\frac{(s^*_j)^2}{\sigma_*^2} = \frac{1}{\sigma_*^2}$. Since the left part of \cref{eq:ana_gaus} is monotonically increasing with $H$, \cref{eq:ana_gaus} will hold for $H^\prime$ if $H^\prime\leq H_*$, namely $\sum_{j}^{m}\frac{s_j^2}{\sigma_j^2} \leq \frac{1}{\sigma_*^2}$.
\end{proof}

\cref{fst_lemma} shows the condition of differential privacy when we add noise from Gaussian distributions with different variances to different dimensions of a query function. For example, suppose we have a $2$-dimension query function $f^\prime(\cdot)$, and $s_1 = 12.0, s_2 = 6.0$. If mechanism $M$ satisfies $(\epsilon, \delta)$-DP with $\sigma_* = 1.0$, then mechanism $M^\prime$ satisfies the same $(\epsilon, \delta)$-DP with $\sigma_1 = 17.0$ and $\sigma_2 = 8.5$ since $\frac{12.0 ^ 2}{17.0 ^ 2} + \frac{6.0 ^ 2}{8.5 ^ 2} \leq \frac{1}{1.0 ^ 2}$. Namely, we sample noise with standard deviation $17.0$ for the first dimension and $8.5$ for the second dimension of $f^\prime(\cdot)$. In contrast, previous methods sample noise from the same Gaussian distribution $\mathcal{N}(0, 13.5^2)$ for each dimension of $f^\prime(\cdot)$ since $\frac{12.0 ^ 2}{13.5 ^ 2} + \frac{6.0 ^ 2}{13.5 ^ 2} \leq \frac{1}{1.0 ^ 2}$. 

Continue with the above example, we conducted a numerical experiment to compare the influence of different noise distributions on the original query function. Specifically, we calculate cosine similarity between the noisy result ($M^\prime(D)$) and the original result ($f^\prime(D)$). As a metric of such influence, the higher similarity two results have, the less the influence of noise on the original query function. Without loss of generality, we assume $f^\prime(D) = (10.0, 5.0)^T$ given that $s_1 = 12.0, s_2 = 6.0$. We then sample noise 10000 times and compute the average cosine similarity. When we set $\sigma_1 = 17.0$ and $\sigma_2 = 8.5$, the average cosine similarity between the noisy result and the original result is 0.52. However, when we set $\sigma_1 = \sigma_2 = 13.5$, the average is only 0.36. Another strategy is to set $\sigma_1 = 12.4$ and $\sigma_2 = 24.8$, the cosine similarity in this scenario is reduced to 0.28. Note that like the third strategy, the optimization techniques in \cite{xiang2019differentially} tend to sample noise with higher variance for dimensions with smaller sensitivity. In other words, this numeric experiment shows that coordinates of the query function with larger sensitivity can tolerate noise with higher variance.

In the scenario of training a deep learning model, we consider the gradient at each iteration as a query function of the training dataset. Specifically, given the gradient $g_t$ at the $t$\textsuperscript{th} iteration, we still use $s_i$ to denote the $\ell^2$-sensitivity of the $i$\textsuperscript{th} dimension of $g_t$. Then depending on \cref{fst_lemma}, a Gaussian mechanism $M_t(x) = g_t(x) + Z$, where $Z=(z_1, \dots, z_m)^T$ and $\forall i\in [m],\, z_i \sim \mathcal{N}(0, \sigma_i^2)$, satisfies $(\epsilon, \delta)$-DP if $\sum_{i=1}^{m}\frac{s_i^2}{\sigma_i^2} \leq \frac{1}{\sigma *^2}$ where $\epsilon, \delta, \sigma_*$ will be determined later. The remained problem is how to calculate $s_i$ and $\sigma_i$. Note that \textsc{RMSProp} uses the term $\expect[g^2]_{t-1}$ to estimate the average square of historical gradients. We observe that historical gradients can also be used to estimate the current gradient. For example, if ${\expect[g^2]_{t-1}^i} > {\expect[g^2]_{t-1}^j} $, then it is quite possible that $|g_t^i| > |g_t^j|$. In other words, $\expect[g^2]_{t-1}$ can be considered as a kind of prior knowledge about the value of $|g_t|$. To distinguish from the term used in the adaptive learning rate, we use $\expect^\prime[g^2]$ to denote this prior knowledge which can be computed as follows:
\begin{align}
	\label{ada_grads_est}
	\expect^\prime[g^2]_0 &= \vv{0}\nonumber\\
	\expect^\prime[g^2]_t &= \gamma^\prime \expect^\prime[g^2]_{t-1} + (1-\gamma^\prime) (g_t)^2
\end{align}

Then depending on the prior knowledge, $s_i$ can be calculated as $\beta \sqrt{\expect^\prime[g^2]^i_{t-1}}$ where $\beta > 0$ is a parameter. Since $s_i$ is the $\ell^2$-sensitivity of the $i$\textsuperscript{th} dimension, we clip each coordinate of the original gradient to guarantee this. Specifically, the $i$\textsuperscript{th} dimension of the clipped gradient (denoted by $\Bar{g}^i_t$) is calculated as:
\begin{align}
	\Bar{g}^i_t = \min\{\max\{g^i_t, {-s_i}\}, {s_i}\}.
\end{align} In contrast with previous methods which clip the gradient with a global $\ell^2$-norm upper bound $C$ such that $\|g_t\|_2 \leq C$, we call our method as \textit{local clipping} since it operates on each coordinate separately and call $\beta$ as \textit{local clipping factor}. Then considering $\sum_{i=1}^{m}\frac{s_i^2}{\sigma_i^2} \leq \frac{1}{\sigma *^2}$, we have $\sigma_i = \beta \sigma_* \sqrt{m\expect^\prime[g^2]_{t-1}^i}$ where $m$ is the number of dimensions. This heuristic means that for the $i$\textsuperscript{th} dimension of the gradient, the standard deviation of Gaussian noise ($\sigma_i$) is scale to the $\ell^2$-sensitivity of the $i$\textsuperscript{th} dimension ($\beta\sqrt{\expect^\prime[g^2]_{t-1}^i}$). On the other hand, as the training tends to converge, the gradual decline of $\sqrt{\expect^\prime[g^2]_{t-1}^i}$ will lead to the decrease of each $\sigma_i$. That means, the noise distributions are not only adaptive to different dimensions of the gradient, but also adaptive to different iterations of the training.

Since we have $\expect^\prime[g^2]_0 = \vv{0}$ at the first iteration, this value cannot be used to clip the gradient which will cause $s_i = 0$. We set another parameter $G$ called \textit{local clipping threshold} and \Alg applies local clipping only when $\Var[\expect^\prime[g^2]] > G$. In other words, as a kind of prior knowledge about the gradient $g_t$, $\expect^\prime[g^2]_{t-1}$ lacks sufficient information at the beginning of the training. During this phase, \Alg applies \emph{global clipping} as follows:
\begin{align}
    \Bar{g}_t^i = \frac{g_t^i}{\max\{1, \|g_t\|_2 / C\}}\, ,
\end{align}
where $C$ is a parameter called \textit{$\ell^2$-norm clipping bound}. After several iterations,  $\expect^\prime[g^2]$ will become more stable and \Alg will perform local clipping when $\Var[\expect^\prime[g^2]] > G$.

To gain more insight into the advantages of adaptive noise, we illustrate and compare the updates of \Alg, \AlgDpSgd, and their non-private counterparts by testing them on the Beale function
$f(x, y) = (1.5-x+xy)^2+(2.25-x+xy^2)^2\text{+}(2.625-x+xy^3)^3$, as shown in \cref{opt_fig:rms,opt_fig:Sgd}. In this experiment, we set $\gamma = 0.1$, $\gamma^\prime=0.9$, $\beta=1.5$ and $G=10^{-6}$.

We observe
that the trajectory of \AlgDpSgd exhibits a remarkable deviation from that of its non-private version, while
\Alg and its non-private version display similar trajectories.
Quantitatively, We define the distance between two trajectories $P=\{p_i\}_{i=1}^n$ and $Q=\{q_i\}_{i=1}^n$ of the same length as
$D(P, Q) = \frac{1}{| P |} \sum_{i=1}^n {\| p_i - q_i\|_2}$. Note that the length of a trajectory is the number of training iterations.
The distance between the trajectories of \AlgDpSgd and its non-private version is 0.90 and the distance between \Alg and its non-private version is 0.21.
Experiments on all widely used test functions for optimization \cite{wiki:Test_functions_for_optimization} also show similar results.

The above experiment suggests that adaptive noise can mitigate the deviation of the noisy result from the original result and the performance of \Alg is more robust to the privacy-protecting Gaussian noise than \AlgDpSgd.
Recall that \AlgDpSgd adds Gaussian noise with the same intensity to all dimensions.
In contrast, \Alg updates each dimension separately and the intensity of the added Gaussian noise relies on
the sensitivity of each dimension of the gradient.

\begin{figure}[!ht]
	\centering
	\begin{subfigure}[b]{0.22\textwidth}
		\includegraphics[width=\textwidth]{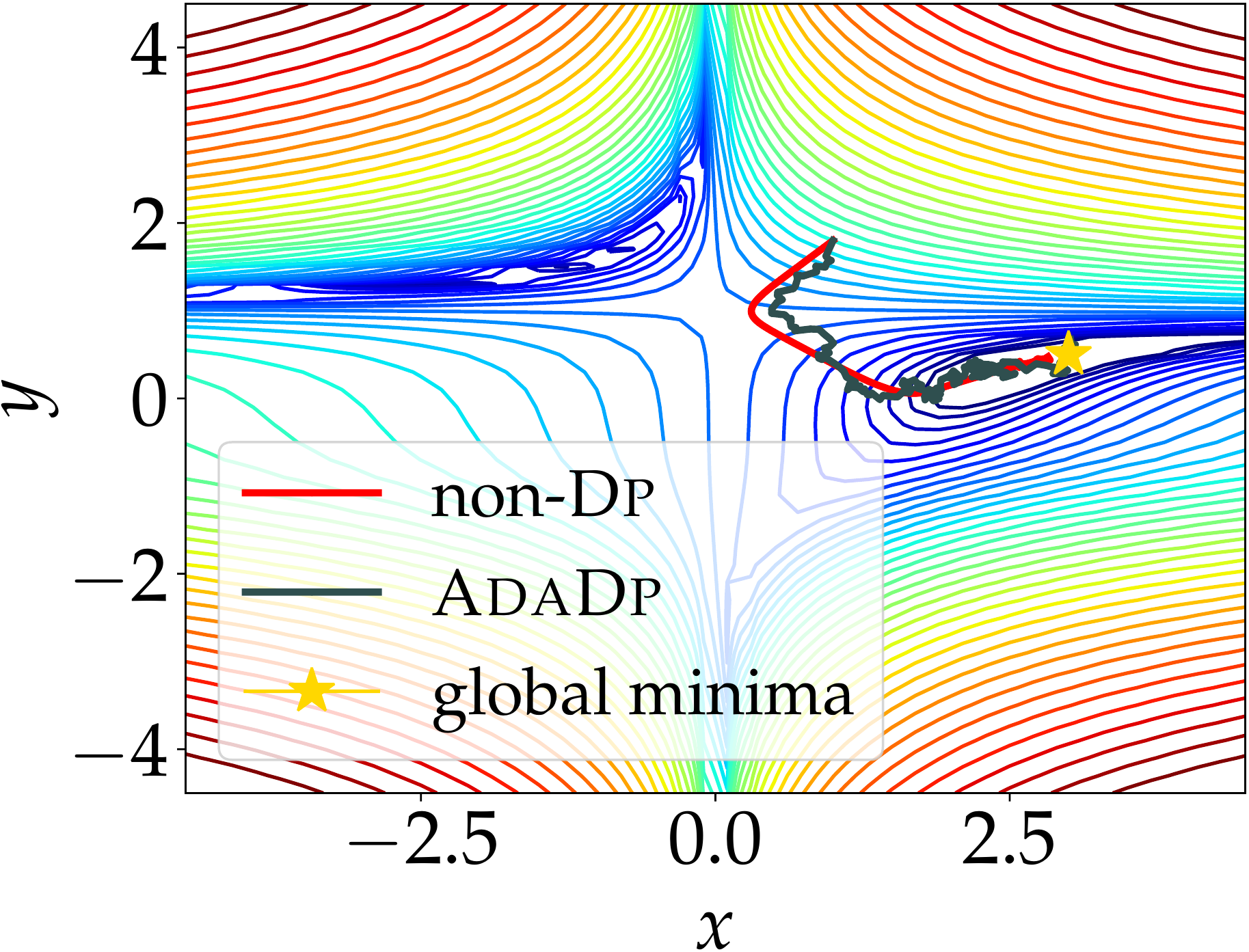}
		\caption{\Alg}
		\label{opt_fig:rms}
	\end{subfigure}
	~
	\begin{subfigure}[b]{0.22\textwidth}
		\includegraphics[width=\textwidth]{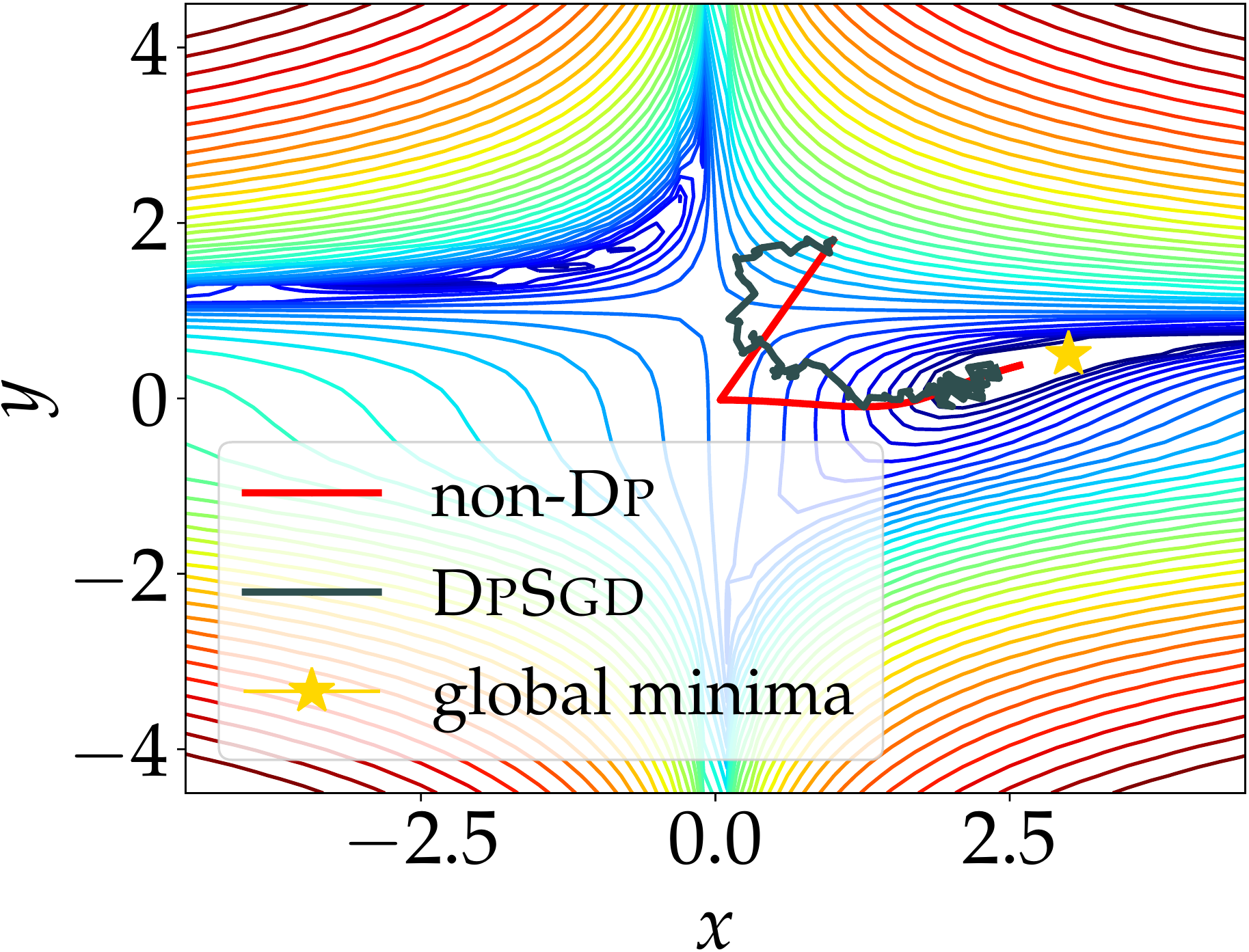}
		\caption{\AlgDpSgd}
		\label{opt_fig:Sgd}
	\end{subfigure}

	\caption{The trajectories of \Alg, \AlgDpSgd and corresponding non-private algorithms on Beale function with 150  training iterations.}\label{fig:opt}
\end{figure}



\subsection{Algorithm and Main Results}
\label{alg_rslt}
\begin{algorithm}[!ht]
	\caption{\Alg}
	\label{the_alg}
	\textbf{Input:} examples ${\{x_1, x_2, \dots, x_N\}}$, learning rate ${\eta}$, lot size ${L}$, noise scale ${\sigma_*}$, local clipping threshold $G$, $\ell^2$-norm clipping bound $C$\\
	\textbf{Output:}  ${\theta_T}$ and privacy cost ${(\epsilon, \delta)}$
	\begin{algorithmic}[1]
		\State \textbf{Initialize} $\theta_0 = \{\theta_0^1, \theta_0^2, \dots, \theta_0^m\}$ randomly
		\For{${t\in[T]}$}
		\State Take a random sample ${L_t}$ with probability ${L/N}$
		\State \textbf{Compute the gradient for each example ${x}$ in $L_t$}
		\State \label{ln:orig_grad} ${g_t(x) \gets \nabla{\mathcal{L}(\theta_{t-1}, x) }}$
		\For{the $i$\textsuperscript{th} dimension of the gradient}
		\If{$\Var[\sqrt{\expect^\prime[g^2]_{t-1}}] > G$} \label{ln:local_clip_bound}	
			\State \textbf{Gradient local clipping}
			\State \label{ln:clip_1}$s_i = \beta\sqrt{{\expect^\prime[g^2]_{t-1}^i}}$
			\State \label{ln:clip_2}${\Bar{g}_t^i \gets \min\{\max\{g_t^i, {-s_i}\}, {s_i}\}}$
			\State \textbf{Add adaptive Gaussian noise}
			\State \label{ln:noise_1}$\sigma_i = \beta\sigma_*\sqrt{m \expect^\prime[g^2]_{t-1}^i}$
			\State \label{ln:noise_2} ${\tilde{g}_t^i \gets \Bar{g}_t^i + \mathcal{N}(0,\, \sigma_i^2)}$
		\Else
			\State \textbf{Gradient global clipping}
			\State \label{ln:global_clip} $\Bar{g}_t^i \gets \frac{g_t^i}{\max\{1, \|g_t\|_2 / C\}}$
			\State\textbf{Add non-adaptive Gaussian noise}
			\State \label{ln:global_noise} ${\tilde{g}_t^i \gets \Bar{g}_t^i + \mathcal{N}(0,\,C^2\sigma_*^2)}$
		\EndIf
		\EndFor
		\State \textbf{Update the expectations}
		\State \label{ln:exp_upd_1}$\expect[\tilde{g}^2]_t \gets (1-\gamma)\expect[\tilde{g}^2]_{t-1} + \gamma (\tilde{g}_t)^2$
		\State \label{ln:exp_upd_2}$\expect^\prime[g^2]_t \gets \gamma^\prime
		\expect^\prime[{g}^2]_{t-1} + (1-\gamma^\prime) (g_t)^2$
		\State \textbf{Update the parameter}
		\State \label{ln:upd_term} ${\Delta\hat{\theta}_t \gets \frac{ \tilde{g}_t}{\sqrt{\expect[\tilde{g}^2]_t + \epsilon_0}}}$
		\State \label{ln:update}${\theta_t \gets \theta_{t-1} - \eta \frac{1}{L}\sum_{i}{\Delta\hat{\theta}_t}}$
		\EndFor
		\State \label{ln:return}\Return $\theta_{T} = \{\theta_T^1, \theta_T^2, \dots, \theta_T^m\}$
	\end{algorithmic}
\end{algorithm}

Now we present  \Alg in \cref{the_alg}. Let $\theta_0 =
\{\theta_0^1, \theta_0^2, \dots, \theta_0^n\}$ denote the initial values of the  $n$
trainable parameters. Note that, a lot is a sample from the training dataset
which is different from batch. There can be several batches in a lot and the lot
size is the batch size multiplied with the number of batches in it.
At each step $t$, we compute the gradient with per-example
gradient algorithm \cite{goodfellow2015efficient} (\cref{ln:orig_grad}). Then for each dimension of the gradient, we perform the local clipping at line \cref{ln:clip_1} and \cref{ln:clip_2} if the local clipping threshold is satisfied (\cref{ln:local_clip_bound}). After that, we calculate $\sigma_i$ and add Gaussian noise to the $i$\textsuperscript{th} dimension of the gradient (\cref{ln:noise_1} and \cref{ln:noise_2}). If the local clipping threshold is not achieved, we perform global clipping and add non-adaptive noise at \cref{ln:global_clip} and \cref{ln:global_noise}. Then we update the two expectations used in the adaptive learning rate and the adaptive noise correspondingly (\cref{ln:exp_upd_1} and \cref{ln:exp_upd_2}). 
Afterwards, we
update the parameters at the $t$\textsuperscript{th} iteration with ${\Delta\hat{\theta}_t}$ (\cref{ln:update}). After $T$ training steps, \Alg returns
 the final values $\theta_{T} = \{\theta_T^1, \theta_T^2, \dots,
\theta_T^m\}$ (\cref{ln:return}).

Before presenting our main results on the privacy guarantee of \Alg, let us review the definition of differential privacy.
\begin{definition} [Differential Privacy \cite{dwork2006calibrating, abadi2016deep}]
A randomized mechanism $M$ satisfies $(\epsilon,\delta)$-differential privacy, if for any two neighboring datasets $D$ and $D^\prime$ that differ only in
one tuple, and for any possible subset of outputs $\mathcal{O}$ of $M$, we have
\begin{align*}
\Pr( M(D) \in \mathcal{O} ) \leq e^{\epsilon} \cdot \Pr(M(D^\prime) \in  \mathcal{O})+\delta,
\end{align*}
where $\Pr( \cdot )$ denotes the probability of an event. If $\delta=0$, $M$ is said to satisfy $\epsilon$-differential privacy.
\end{definition}


We now show the privacy guarantee of \Alg in \cref{sec_threm}.

\begin{theorem}
\label{sec_threm}
Let $\binom{\cdot}{\cdot}$ denote the binomial coefficient and $B(l) =  \sum_{i=0}^{l}(-1)^i\binom{l}{i}e^{(i-1){i}/{(2\sigma_*^2)}}$.
 Given a privacy budget $(\epsilon, \delta)$, the sampling ratio $q=\frac{L}{N}$ and any integer $\alpha \geq 2$, if the noise scale $\sigma_i$ in Algorithm \ref{the_alg} satisfies $\sum_{i}^{m}\frac{s_i^2}{\sigma_i^2} \leq \frac{1}{\sigma_*^2}$ and $\sigma_*$ satisfies
	\begin{align} \label{threm_cond}
	    &\frac{T}{\alpha - 1} \log\Big(1 + q^2\binom{\alpha}{2}\min\{4(e^{1/\sigma_*^2} - 1), 2e^{1/\sigma_*^2}\} \nonumber \\ & +  4\sum_{j=3}^{\alpha}q^j\binom{\alpha}{j}\sqrt{B(2\ceil{j/2})\cdot B(2\floor{j/2})}\Big)+ \frac{\log(1/\delta)}{\alpha - 1} \leq \epsilon,
	\end{align}
 then Algorithm \ref{the_alg} is $(\epsilon, \delta)$-differentially private.
\end{theorem}

To prove \cref{sec_threm}, we use the techniques of RDP to analyze the privacy cost of the composition of Gaussian mechanisms. The results that we obtain via RDP are later translated to DP.

\begin{definition} [R\'enyi Divergence \cite{mironov2017renyi}]
Given two probability distributions $P$ and $Q$, the R\'enyi divergence between $P$ and $Q$ with order $\alpha > 1$ is defined by
	$$D_\alpha(P||Q) = \frac{1}{\alpha - 1}\log\mathbb{E}_{x\sim Q}\Big(\frac{P(x)}{Q(x)}\Big)^\alpha.$$
\end{definition}

\begin{definition} [R\'enyi differential privacy \cite{mironov2017renyi}]
 A randomized mechanism $f: D \rightarrow \mathbb{R}^d$ is said to be $(\alpha, \epsilon)$-R\'enyi differential private if for any adjacent datasets $D, D^\prime$, it holds that
	$$D_\alpha(f(D)||f(D^\prime)) \leq \epsilon.$$
\end{definition}

Our proofs are based on the three key properties of RDP, as stated in  \cref{apdx_lem3,apdx_lem1,apdx_lem2}.
\begin{lemma}[RDP of Gaussian Mechanism \cite{mironov2017renyi}]
	\label{apdx_lem3}If $\|f(\cdot)\|_2 \leq 1$, then the Gaussian mechanism $M(D) = f(D) + \mathcal{N}(0, \sigma^2)$ satisfies $(\alpha, \alpha/(2\sigma_2))$-RDP.
\end{lemma}

\begin{lemma} [Composition of RDP \cite{mironov2017renyi}]
	\label{apdx_lem1}
     For two randomized mechanisms $f, g$ such that $f$ is $(\alpha, \epsilon_1)$-RDP and $g$ is $(\alpha, \epsilon_2)$-RDP, the composition of $f$ and $g$ which is defined as $(X, Y)$ (a sequence of results), where $X\sim f$ and $Y\sim g$, satisfies $(\alpha, \epsilon_1 + \epsilon_2)$-RDP.
\end{lemma}
\begin{lemma} [Translation from RDP to DP \cite{mironov2017renyi}]
\label{apdx_lem2}
If a randomized mechanism $f: D \rightarrow \mathbb{R}^d$ satisfies $(\alpha, \epsilon)$-RDP, then it satisfies $(\epsilon + \frac{\log1/\delta}{\alpha - 1}, \delta)$-DP where $0 <\delta <1$.
\end{lemma}

\cref{apdx_threm} analyzes the privacy amplification in the subsampling setting.
\begin{lemma} [RDP of subsampled Gaussian Mechanism \cite{wang2018subsampled}]
	\label{apdx_threm}
     Define function $B(\epsilon, l)= \sum_{i=0}^{l}(-1)^i\binom{l}{i}e^{(i-1)\epsilon(i)}$. Given a dataset $D$ of $n$ records and a Gaussian mechanism $f$ satisfying $(\alpha, \epsilon(\alpha))$-RDP,
define a new randomized mechanism $f \circ \mathbf{subsample}$ as: (1) subsample $m$ records where $ q = m/n$, (2) apply these $m$ records as the input of mechanism $f$, then for any integer $\alpha > 1$, $f \circ \mathbf{subsample}$ satisfies $(\alpha, \epsilon^\prime(\alpha))$-RDP, where:
	\begin{align}
	\label{sub_rdp}
\epsilon^\prime(\alpha) \leq &\dfrac{1}{\alpha - 1} \log\Big(1 + q^2\binom{\alpha}{2}\min\{4(e^{1/\sigma^2 - 1}), 2e^{1/\sigma^2}\} \nonumber \\ & + 4\sum_{j=3}^{\alpha}q^j\binom{\alpha}{j}\sqrt{B(2\ceil{j/2})\cdot B(2\floor{j/2})}\Big).
\end{align}
\end{lemma}

Before showing the proof of \cref{sec_threm}, we need to establish the following important lemma. In \cref{fst_threm}, we bound the noise level $\sigma$ of a Gaussian mechanism in the subsampling setting.
\begin{lemma}
\label{fst_threm}
	Given the sampling probability $q = L/N$, the number of steps $T$, the Gaussian mechanism $M = f(\cdot) + \mathcal{N}{(0, \sigma^{2})}$ where $\| f(\cdot)$ $\|_2 \leq 1$, then the composition of $T$ these mechanisms satisfies $(\epsilon, \delta)$-differentially private if $\sigma$ satisfies:
	\begin{align}
	\label{eq:fst_threm}
	&\dfrac{T}{\alpha - 1} \log\Big(1 + q^2\binom{\alpha}{2}\min\{4(e^{1/\sigma^2} - 1), 2e^{1/\sigma^2}\} \\ & + 4\sum_{j=3}^{\alpha}q^j\binom{\alpha}{j}\sqrt{B(2\ceil{j/2})\cdot B(2\floor{j/2})}\Big) + \dfrac{\log(1/\delta)}{\alpha - 1} \leq \epsilon,
	\end{align}
	where $\alpha$ can be any integer satisfying $\alpha \geq 2$ and the function $B$ is defined as $B(l) = \sum_{i=0}^{l}(-1)^i\binom{l}{i}e^{(i-1){i}/{(2\sigma^2)}}$.
\end{lemma}

\begin{proof}
	By {Lemma \ref{apdx_lem3}}, we could compute $B(\epsilon, l)$ defined in {Lemma \ref{apdx_threm}} on Gaussian mechanisms as
	\begin{align}
		\label{func_b}
		B(l) = \sum_{i=0}^{l}(-1)^i\binom{l}{i}e^{(i-1){i}/{(2\sigma^2)}}.
	\end{align}
	 Since a lot used in \Alg is a subsample of the training dataset, each step is $(\alpha, \epsilon^\prime(\alpha))$-RDP where $\epsilon^\prime(\alpha)$ satisfies \eqref{sub_rdp} in {Lemma \ref{apdx_threm}}. Then after $T$ training steps of \Alg, we could obtain the total privacy cost $(\alpha, T\epsilon^\prime(\alpha))$ via composing such $T$ subsampled Gaussian mechanisms depending on {Lemma \ref{apdx_lem1}}. Then by substituting the function $B(\epsilon, l)$ in \eqref{sub_rdp} with  \eqref{func_b}, $T\epsilon^\prime(\alpha)$ can be clearly expressed as
	 \begin{align}
	 \label{total_rdp}
	 	T\epsilon^\prime(\alpha) \leq &\frac{T}{\alpha - 1}\log\Big(1 + q^2\binom{\alpha}{2}\min\{4(e^{1/\sigma^2} - 1), 2e^{1/\sigma^2}\} \nonumber\\ & +  4\sum_{j=3}^{\alpha}q^j\binom{\alpha}{j}\sqrt{B(2\ceil{j/2})\cdot B(2\floor{j/2})}\Big).
	 \end{align}
	 At last, through converting $(\alpha, T\epsilon^\prime(\alpha))$ to DP representation via {Lemma \ref{apdx_lem2}}, \Alg satisfies $(T\epsilon^\prime(\alpha)+\frac{\log(1/\delta)}{\alpha - 1}, \delta)$-DP. Let $T\epsilon^\prime(\alpha)+\frac{\log(1/\delta)}{\alpha - 1} \leq \epsilon$ where $\epsilon$ is the given privacy budget and combine this with \eqref{total_rdp} as follows:
	 \begin{align*}
	    &\frac{T}{\alpha - 1} \log\Big(1 + q^2\binom{\alpha}{2}\min\{4(e^{1/\sigma^2} - 1), 2e^{1/\sigma^2}\} \\ & +  4\sum_{j=3}^{\alpha}q^j\binom{\alpha}{j}\sqrt{B(2\ceil{j/2})\cdot B(2\floor{j/2})}\Big)+ \frac{\log(1/\delta)}{\alpha - 1} \leq \epsilon,
	\end{align*}
	 then the proof is completed.
\end{proof}

Now we are ready to prove \cref{sec_threm}.

\begin{proof}


Combining \cref{fst_lemma} and \cref{fst_threm}, \cref{sec_threm} is easy to prove. Suppose we use $M(D) = f(D) + Z$ as a Gaussian mechanism at training step $t$ in \cref{the_alg}, then the composition of such $T$ mechanisms satisfies \cref{eq:fst_threm}. Based on \cref{fst_lemma}, if we use $M^\prime(D) = f^\prime(D) + Z^\prime$ at each training step $t$ in which $\sum_{i=1}^{m}\frac{s_i^2}{\sigma_i^2} \leq \frac{1}{\sigma_*^2}$, then $M^\prime(D)$ satisfies the same differential privacy guarantee as $M(D)$. Therefore, the proof of \cref{fst_threm} is also suitable for $M^\prime(D)$. Combine this fact with the post-processing property of DP \cite{dwork2014algorithmic}, the proof is completed.

\end{proof}

We would like to remark that \cref{sec_threm} covers the realm of small noise and high sampling ratio that the moments accountant method \cite{abadi2016deep} omits (which requires $\sigma \geq 1$ and $q < \frac{1}{16\sigma}$). For instance, if we train a model using \Alg with $L=600$ and $\sigma_*=0.9$ on a dataset of $N=60000$ examples, \cref{sec_threm} bounds the privacy cost by choosing the optimal $\alpha=6$ and implies that the model achieves $(4.0, 10^{-5})$-differential privacy after 1800 training steps.

\section{Experimental Results} \label{sec:experiments}

\subsection{Evaluation Setup}
We evaluate the privacy cost, the accuracy and the computational efficiency of \Alg compared with state-of-the-art methods: \AlgDpSgd \cite{abadi2016deep}, \textsc{AGD} \cite{lee2018concentrated} and \textsc{DpOpt} \cite{xiang2019differentially} on two real datasets: MNIST and
CIFAR-10.
We implemented all these algorithms using TensorFlow
\cite{abadi2016tensorflow} with a GTX 1080Ti GPU.

MNIST is a standard dataset for handwritten digit recognition, which consists of 60,000 training examples and 10,000 testing examples. Each example is a $28\times 28$  gray-level image. CIFAR-10 consists of 60,000 labeled examples of $32\times 32$  RGB images. There are 50,000 training images and 10,000 testing images.

The model for MNIST task first performs a 60-dimensional differentially private PCA (\textsc{Dp}PCA) projection \cite{dwork2014analyze} and then applies a single
1,000-unit ReLU hidden layer \cite{nair2010rectified}.
For the CIFAR-10 task, we use a variant of AlexNet \cite{krizhevsky2012imagenet}, which contains
two convolutional layers followed by two fully connected
layers and one softmax layer.
The first convolutional layer includes 64 filters of size $5\times 5$ with stride 1. The layer is followed by the ReLU activation, $2\times 2$
max pooling, and local response normalization. The structure of the second convolutional layer
is identical to that of the first one except that the local response normalization is performed before max pooling. The
output is then flattened into a
vector as the input for the following fully connected layer.


As outlined in \cref{the_alg}, we compute gradients
for each training example. However,
it is prohibitive to compute per-example gradients due to the parameter sharing scheme of
 convolutional layers.
Since convolutional layers are shown to be well
transferred \cite{jarrett2009best},
we pre-train the model on CIFAR-100 and initialize the network with the trained parameters. When we train it on CIFAR-10, the parameters of convolutional layers are maintained and updates happen to the fully connected layers and the softmax layers.

We use the result of \cref{sec_threm}
to calculate the privacy cost. Specifically, given $\epsilon$, $\sigma_*$, and $q$, at each iteration, we select $\alpha$'s from $\{2, 3,\dots, 64\}$ and determine the smallest $\delta_*$ that satisfies \eqref{threm_cond} in Theorem \ref{sec_threm}. The privacy cost is the pair $(\epsilon, \delta_*)$. 

For \Alg, we set $\gamma = 0.1$, $\gamma^\prime = 0.9$ and $G=10^{-6}$ in all experiments. For \textsc{AGD}, since it was only evaluated on shallow machine learning tasks in \cite{lee2018concentrated}, its privacy computation is not suitable for deep learning as it does not consider the privacy amplification due to subsampling. Therefore, we use \cref{sec_threm} to compute the privacy guarantee for \textsc{AGD} to give a fair comparison. To implement \textsc{DpOpt}\footnote{In the equation that precedes (10a) in~\cite{xiang2019differentially}, the quadratic term in the moment generating function of a Gaussian distribution was missing. This equation should be $\sum_{i}^{d}\frac{\Delta_i^2}{\sigma_i^2} \leq \frac{2\epsilon}{1 + \lambda} + \frac{2}{\lambda + \lambda^2}\ln{\delta}$.}, we use the projected gradient descent algorithm and optimize the objective function with 50 steps as the same as \cite{xiang2019differentially}.

\subsection{Privacy Cost} \label{subsec:privacy_cost}
\begin{table}[!ht]
\centering
	\begin{tabular}[t]{ccccc}
		\toprule
		Dataset & Accuracy& $\delta_*$& $\epsilon_D$ & $\epsilon_A$ \\
		\midrule
		\multirow{4}{*}{MNIST}&0.92       & \multirow{4}{*}{$10^{-4}$}  & 1.10       & 0.75 \\
		&0.94       &  & 1.61       & 0.78 \\
		&0.95       &   & 1.96       & 0.80  \\
		&0.96       &   & 5.48       & 1.40  \\
		\midrule
		\multirow{4}{*}{CIFAR-10}&0.68       & \multirow{4}{*}{$10^{-4}$}  & 0.89       & 0.62 \\
		&0.69       &  & 1.85       & 1.07 \\
		&0.70       &   & 4.75       & 2.13  \\
		&0.71       &   & 8.36      & 3.55  \\
		\bottomrule
	\end{tabular}
	\caption{Results on the privacy cost of \AlgDpSgd and \textsc{AdaDp}}
	\label{the_tab}
\end{table}

To illustrate the trade-off between the privacy cost and the accuracy,
we measured the privacy cost of \Alg and \AlgDpSgd to attain a pre-specified accuracy level. 
We set the noise level $\sigma_* = 3.0$ and $ \sigma_{p}=6.0$ on MNIST and $\sigma_* = 4.0$ on CIFAR-10, where $\sigma_p$ is the noise level for \textsc{DpPCA} \cite{dwork2014analyze}. 

\emph{Compared with \AlgDpSgd (as a representative method with a non-adaptive learning rate), \Alg achieves an average reduction of 54\% and 46\% in privacy cost on MNIST and CIFAR-10 respectively.} \cref{the_tab} summarizes the results, where $\epsilon_D$ and $\epsilon_A$ denote the minimum privacy cost required by \AlgDpSgd and \Alg, respectively, to attain the pre-specified accuracy level. We observe that \Alg always requires much lower privacy cost than \AlgDpSgd to achieve the same accuracy level.
This is mainly due to the faster convergence and fewer training steps of \Alg.

\subsection{Accuracy}

\begin{figure*}[!ht]
	\begin{subfigure}[b]{0.32\textwidth}
		\includegraphics[width=\textwidth]{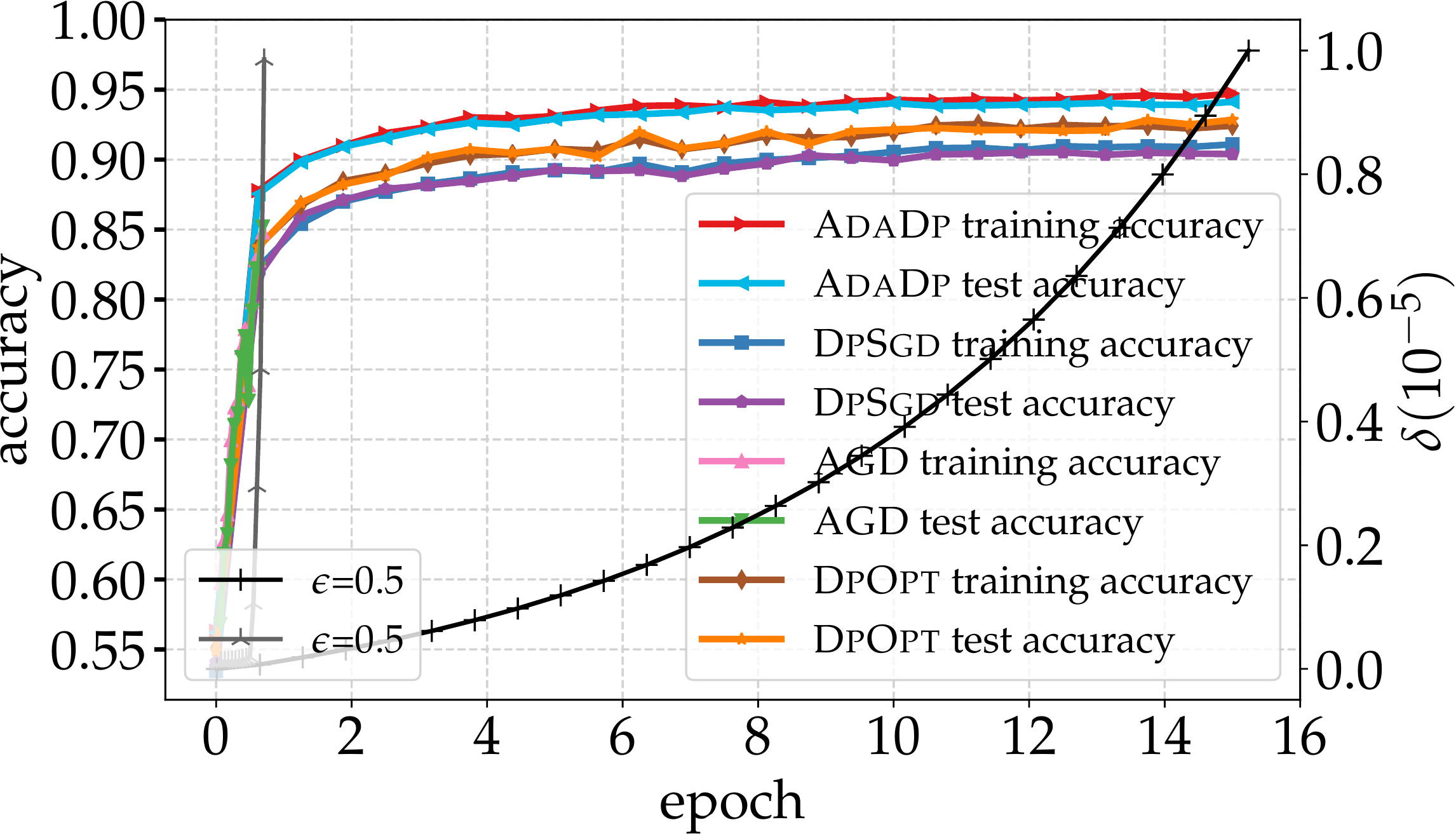}
		\caption{high privacy level}\label{fig:mnist_acc_large_noise}
	\end{subfigure}
	~
	\begin{subfigure}[b]{0.32\textwidth}
		\includegraphics[width=\textwidth]{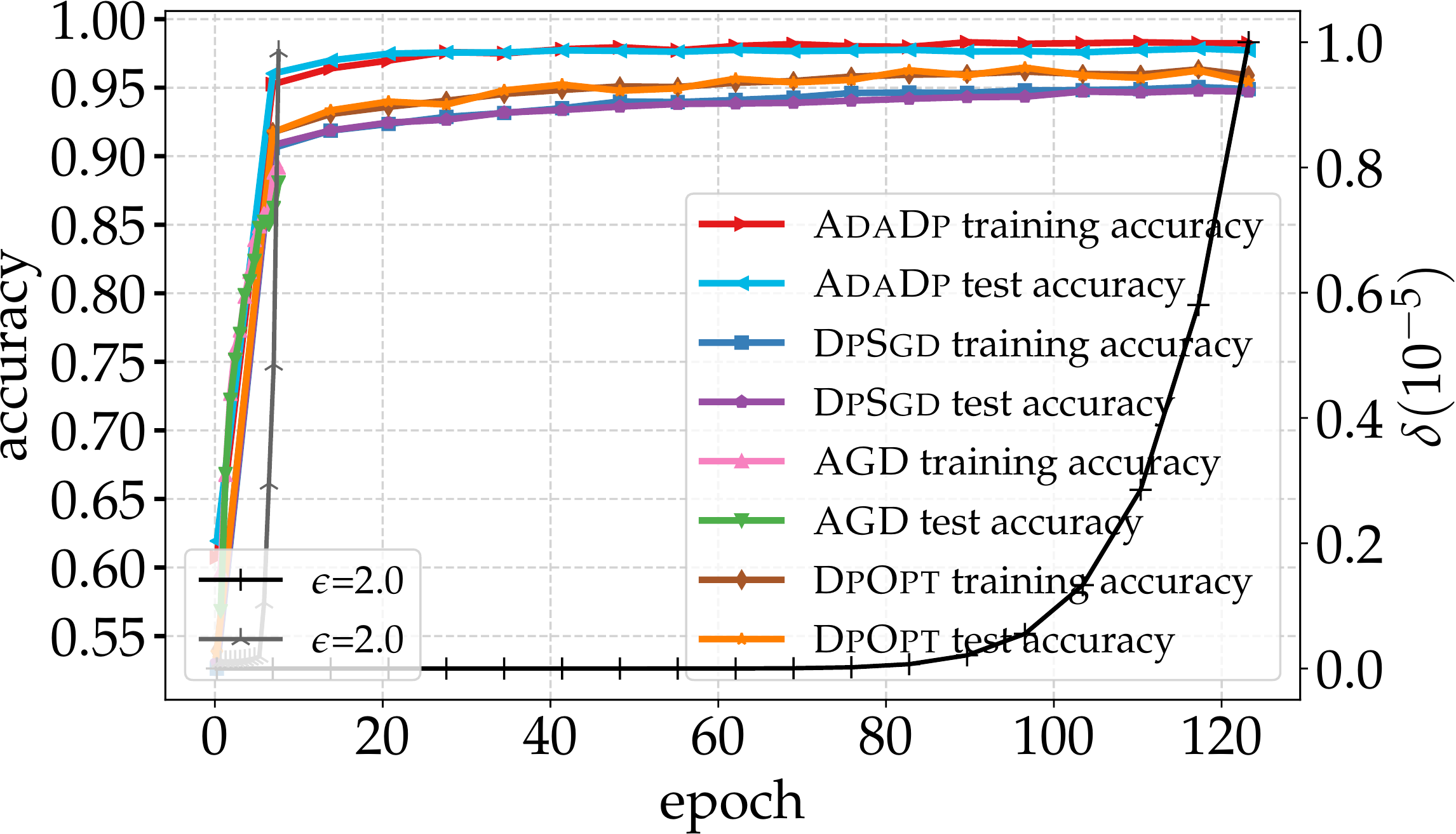}
		\caption{medium privacy level}\label{fig:mnist_acc_median_noise}
	\end{subfigure}
	~
	\begin{subfigure}[b]{0.32\textwidth}
		\includegraphics[width=\textwidth]{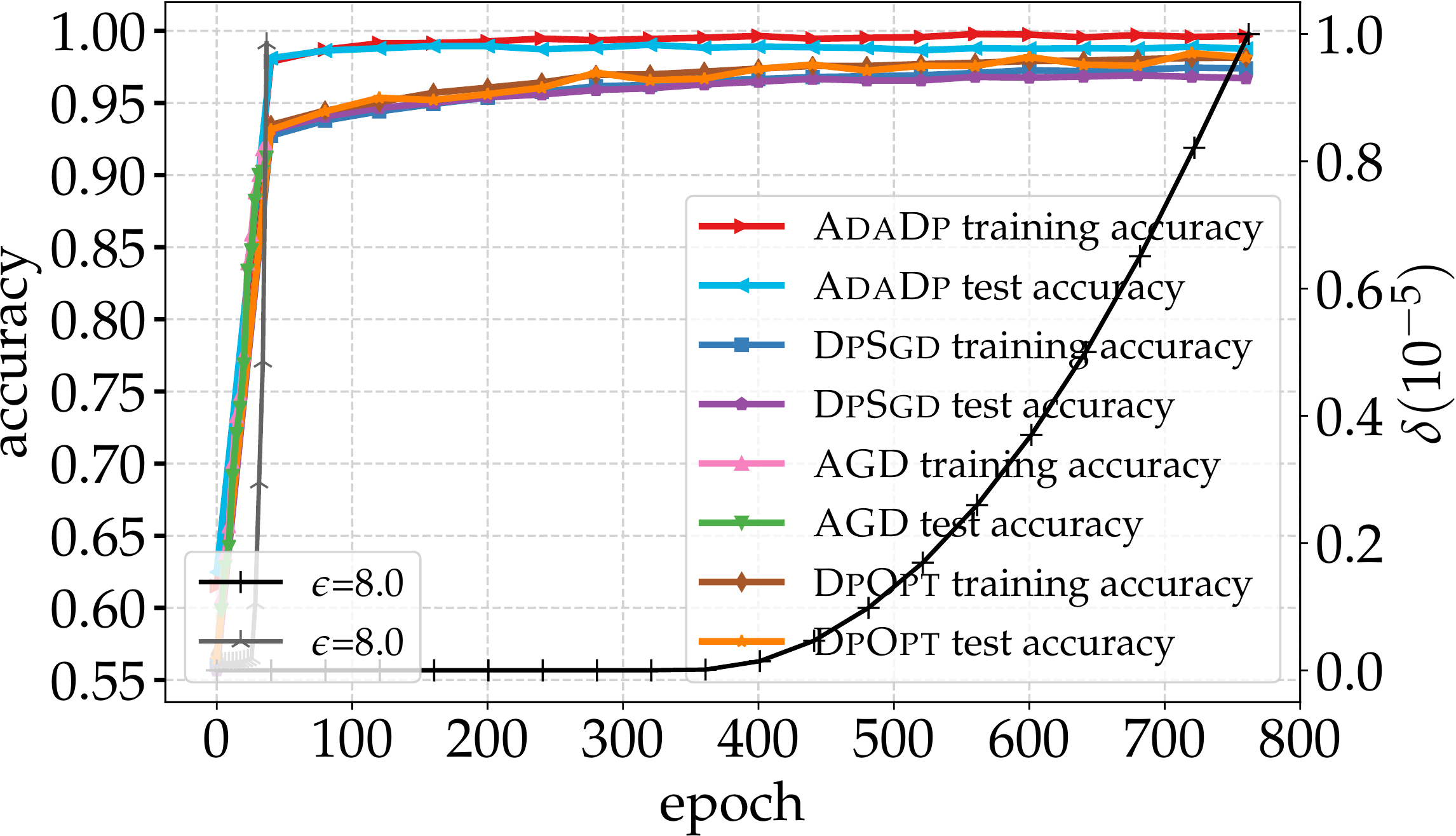}
		\caption{low privacy level}\label{fig:mnist_acc_small_noise}
	\end{subfigure}

	\caption{Results on the accuracy for different privacy levels on MNIST dataset.}\label{fig:mnist_acc}
\end{figure*}
\begin{figure*}[!ht]
	\begin{subfigure}[b]{0.32\textwidth}
		\includegraphics[width=\textwidth]{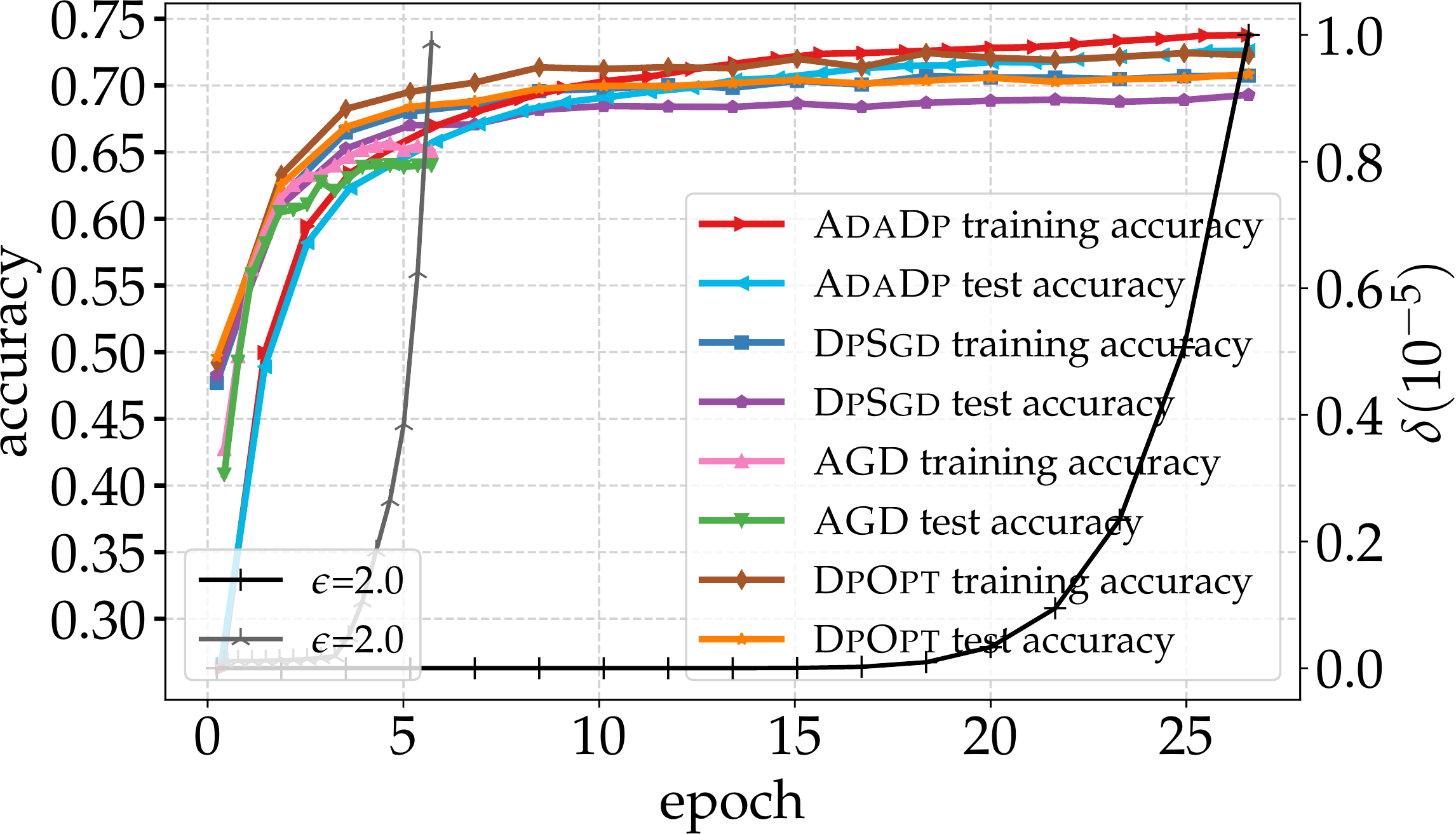}
		\caption{high privacy level}
		\label{fig:cifar_a}
	\end{subfigure}
	~
	\begin{subfigure}[b]{0.32\textwidth}
		\includegraphics[width=\textwidth]{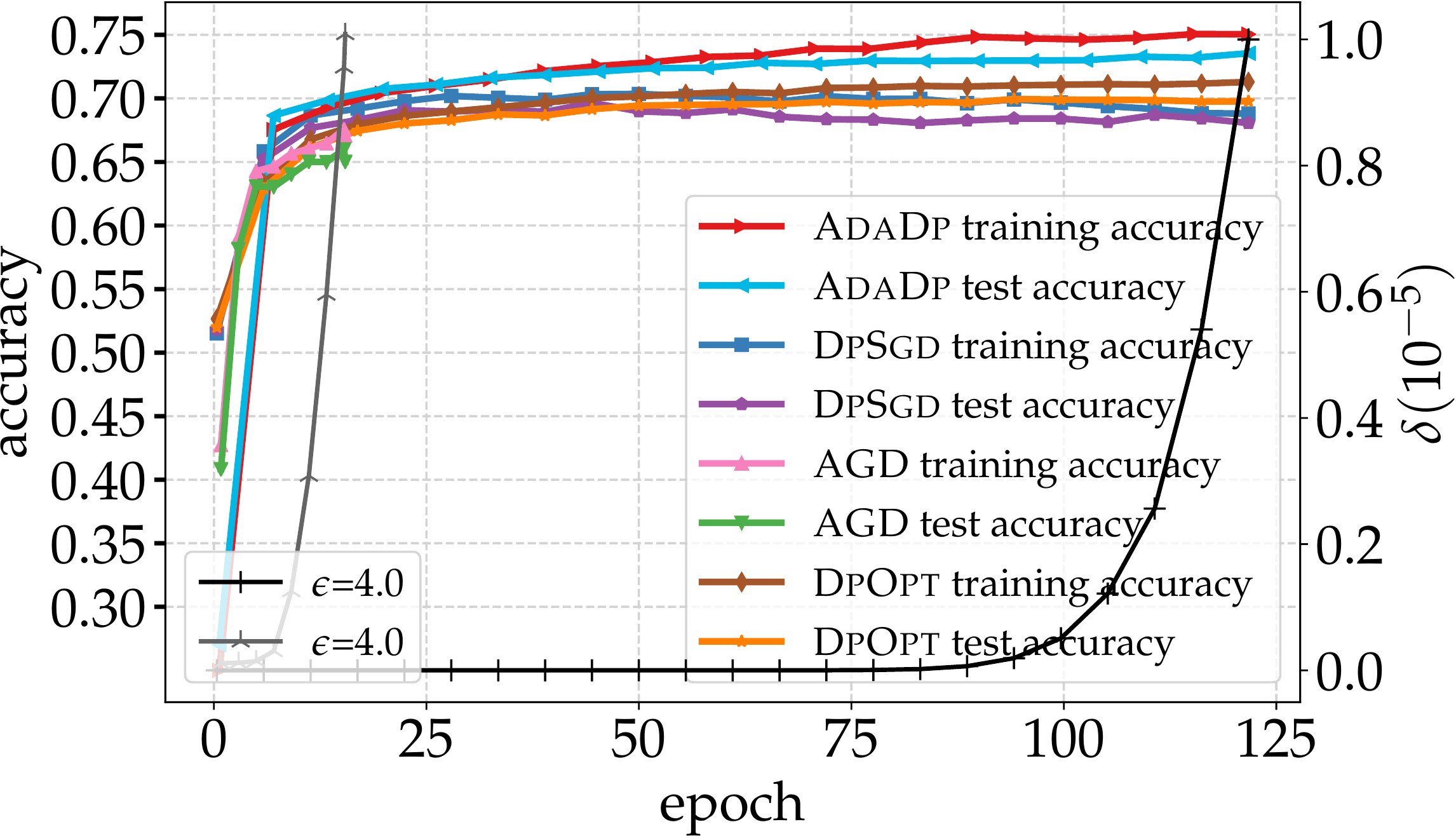}
		\caption{medium privacy level}
		\label{fig:cifar_b}
	\end{subfigure}
	~
	\begin{subfigure}[b]{0.32\textwidth}
		\includegraphics[width=\textwidth]{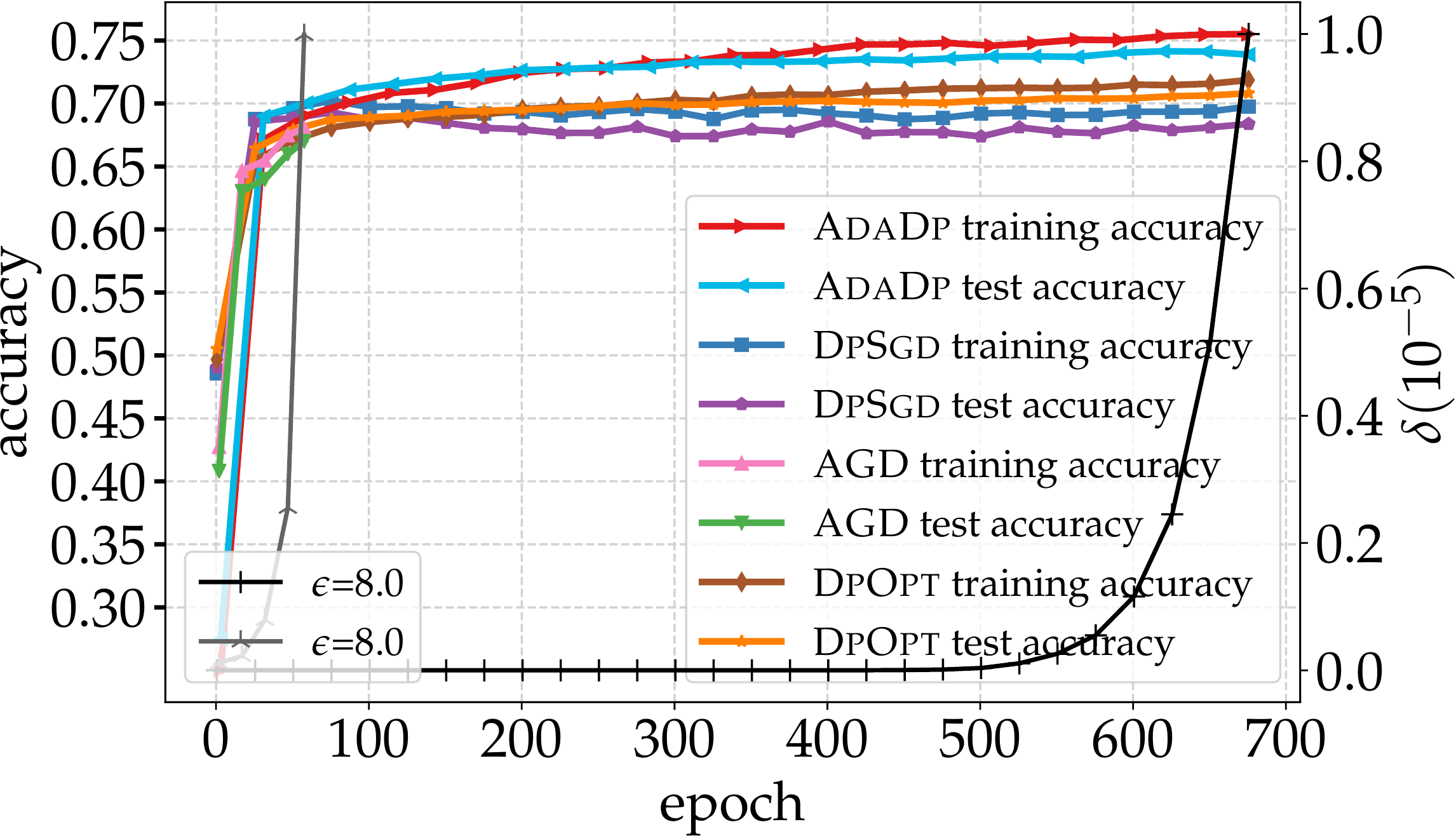}
		\caption{low privacy level}
		\label{fig:cifar_c}
	\end{subfigure}

	\caption{Results on the accuracy for different privacy levels on CIFAR-10 dataset.}
\label{fig:cifar10}
\end{figure*}

Furthermore, given a fixed privacy budget, we evaluate the accuracy of \Alg, \AlgDpSgd, \textsc{AGD} and \textsc{DpOpt} on both MNIST and CIFAR-10 datasets under the high privacy level, the medium privacy level and the low privacy level respectively. In all experiments on MNIST, we set the lot size $L=600$ and $\ell^2$-norm clipping bound $C=4.0$. For \Alg, we set the local clipping factor $\beta=1.2$. The noise levels $(\sigma_*, \sigma_{p})$ for training the fully connected layer and the \textsc{Dp}PCA layer are set to (8, 16), (4, 7), and (2, 4), respectively, for the three experiments. When testing \AlgDpSgd and \textsc{DpOpt}, we set the initial learning rate to 0.1 and linearly drop it down to 0.052 after 10 epochs and then keep it constant. When applying \Alg, we set the learning rate $\eta$ to 0.002 and keep it unchanged. On CIFAR-10, we set the noise level $\sigma_* = 6.0$, lot size $L=2000$, $\ell^2$-norm clipping bound $C=3.0$ for fully connected layers and $C=6.0$ for the softmax layer. For \Alg, we set the local clipping factor $\beta=1.2$. For \AlgDpSgd and \textsc{DpOpt}, we set the initial learning rate to 0.1 and apply an exponentially decaying schedule in the first experiment and set the learning rate to 0.052 in the last two experiments. For \Alg, the learning rate $\eta$ is set to 0.001.


\emph{The accuracy achieved by \Alg significantly outperforms that of \AlgDpSgd, \textsc{AGD} and \textsc{DpOpt} under all three privacy levels on both MNIST and CIFAR-10 datasets.}  \cref{fig:mnist_acc} illustrates how the accuracy varies with the number of epochs and $\delta_*$ on MNIST. The gray line and the black line (please refer to the right vertical axis) denote the accumulating privacy cost for \textsc{AGD} and other methods respectively in terms of $\delta_*$ given a fixed $\epsilon$. We observe that the final test accuracy of \Alg on the MNIST achieves an increase of $5.9\%$, $2.7\%$ and $1.0\%$ respectively compared with \AlgDpSgd, $8.8\%$, $9.5\%$ and $6.7\%$ respectively compared with \textsc{AGD}, $4.3\%$, $1.5\%$ and $0.7\%$ respectively compared with \textsc{DpOpt}. 
The results on CIFAR-10 are shown in \cref{fig:cifar10}. In all three settings, \Alg achieves an accuracy increase of 3.2\%, 5.3\% and 4.8\% respectively compared with \AlgDpSgd, $8.5\%$, $8.0\%$ and $6.9\%$ respectively compared with \textsc{AGD}, $2.1\%$, $3.5\%$ and $3.2\%$ respectively compared with \textsc{DpOpt}. 

We now analyze how we achieve higher accuracy than each state-of-the-art method respectively. Compared with \AlgDpSgd, the performance improvement of \Alg is achieved by both the adaptive learning rate and the adaptive noise. As to \textsc{AGD}, $\sigma$ in this algorithm can only decrease gradually which cannot be reversed such that the privacy budget is consumed too fast. Generally, a few epochs of training are not enough to guarantee the convergence for deep learning models. Therefore, \textsc{AGD} performs poorly in deep learning. For \textsc{DpOpt}, on the one hand, \Alg adopts an adaptive learning rate to improve the convergence. On the other hand, as our first numeric experiment indicates, \Alg samples noise with higher variance for dimensions of the gradient with larger sensitivity which mitigates the influence of noise significantly. In addition, the expected variance of each noise distribution will gradually decrease as the training progresses.

\subsection{Computational Efficiency}
\begin{table}[!ht]
\centering
	\begin{tabular}[t]{ccccc}
		\toprule
		Dataset & \Alg& \AlgDpSgd & \textsc{AGD} & \textsc{DpOpt} \\
		\midrule
		MNIST & 0.5s & 0.3s & 12.0s & 15.5s \\
		\midrule
		CIFAR-10 & 1.2s & 0.9s & 23.7s & 19.1s \\
		\bottomrule
	\end{tabular}
	\caption{Results on the computational efficiency.}
	\label{the_tab_2}
\end{table}
In addition, we evaluate the average processing time of \Alg, \AlgDpSgd, \textsc{AGD} and \textsc{DpOpt} per iteration, as a measure of computational efficiency. The time involves processing a whole batch where we set the batch size as 120 on MNIST and 32 on CIFAR-10 respectively. Other settings are the same as \cref{fig:mnist_acc_large_noise} and \cref{fig:cifar_a} respectively.

\emph{\Alg is much more computationally efficient than both \textsc{AGD} and \textsc{DpOpt}.} The results are shown in \cref{the_tab_2}. We observe that for each iteration, the average processing time of \Alg is close to that of \AlgDpSgd which is far shorter than the other two algorithms. Note that \textsc{AGD} repeatedly evaluates the model multiple times to obtain the best update step size while \textsc{DpOpt} needs to solve a non-convex optimization problem with variables as many as the parameters of the deep learning model. In contrast to \textsc{AGD} and \textsc{DpOpt}, \Alg adopts a heuristic to avoid heavy computation at each iteration.

\subsection{Micro Benchmarks}

\begin{figure*}[!ht]
	\centering
	\begin{subfigure}[b]{0.23\textwidth}
		\includegraphics[width=\textwidth, height=0.115\textheight]{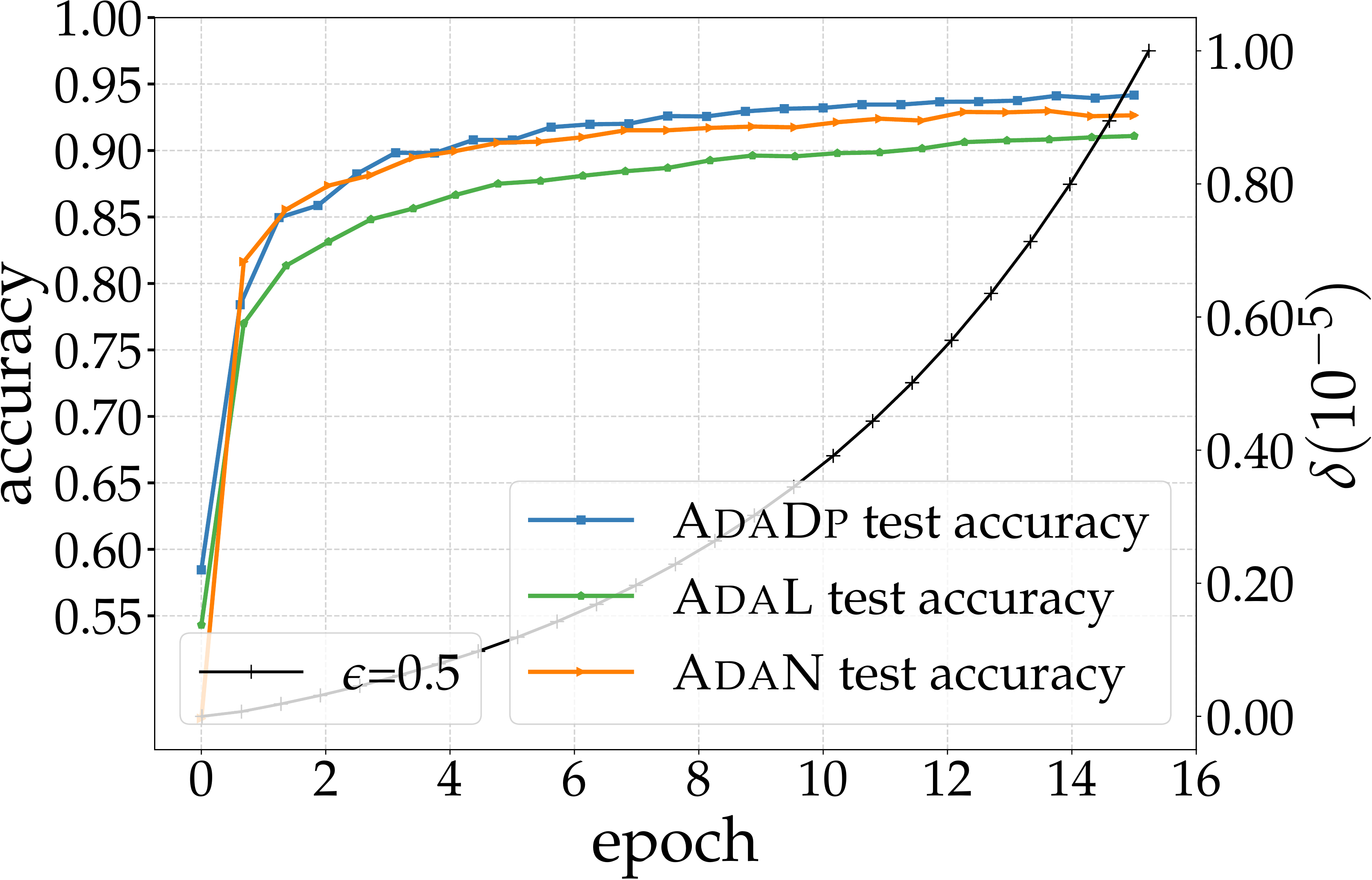}
		\caption{Results on the test accuracy on MNIST dataset.}
		\label{fig:ind_mnist_acc_large_noise}
	\end{subfigure}
	~
	\begin{subfigure}[b]{0.23\textwidth}
		\includegraphics[width=\textwidth, height=0.115\textheight]{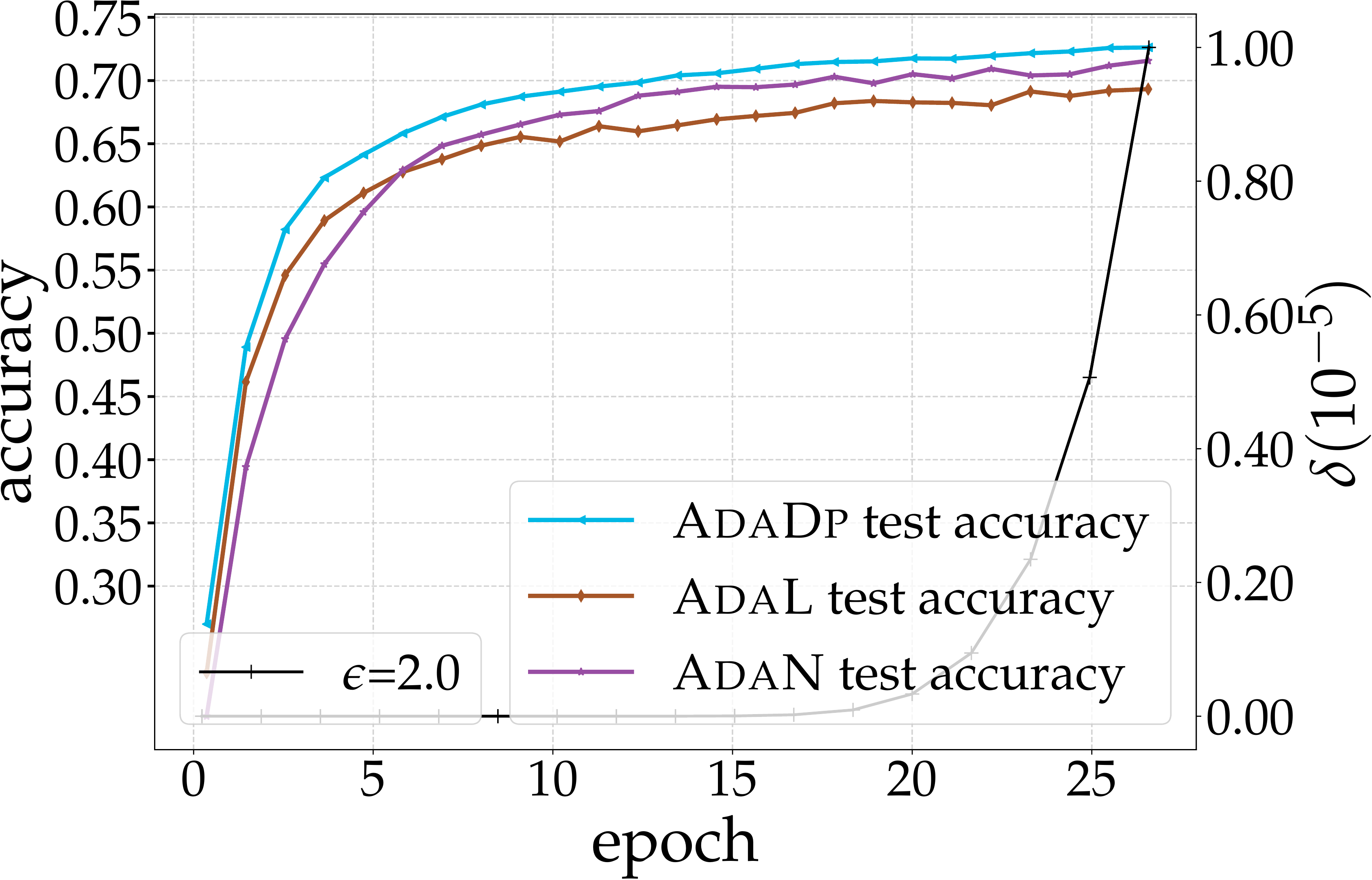}
		\caption{Results on the test accuracy on CIFAR-10 dataset.}
		\label{fig:ind_cifar_acc_large_noise}
	\end{subfigure}
	~
	\begin{subfigure}[b]{0.23\textwidth}
		\includegraphics[width=\textwidth, height=0.115\textheight]{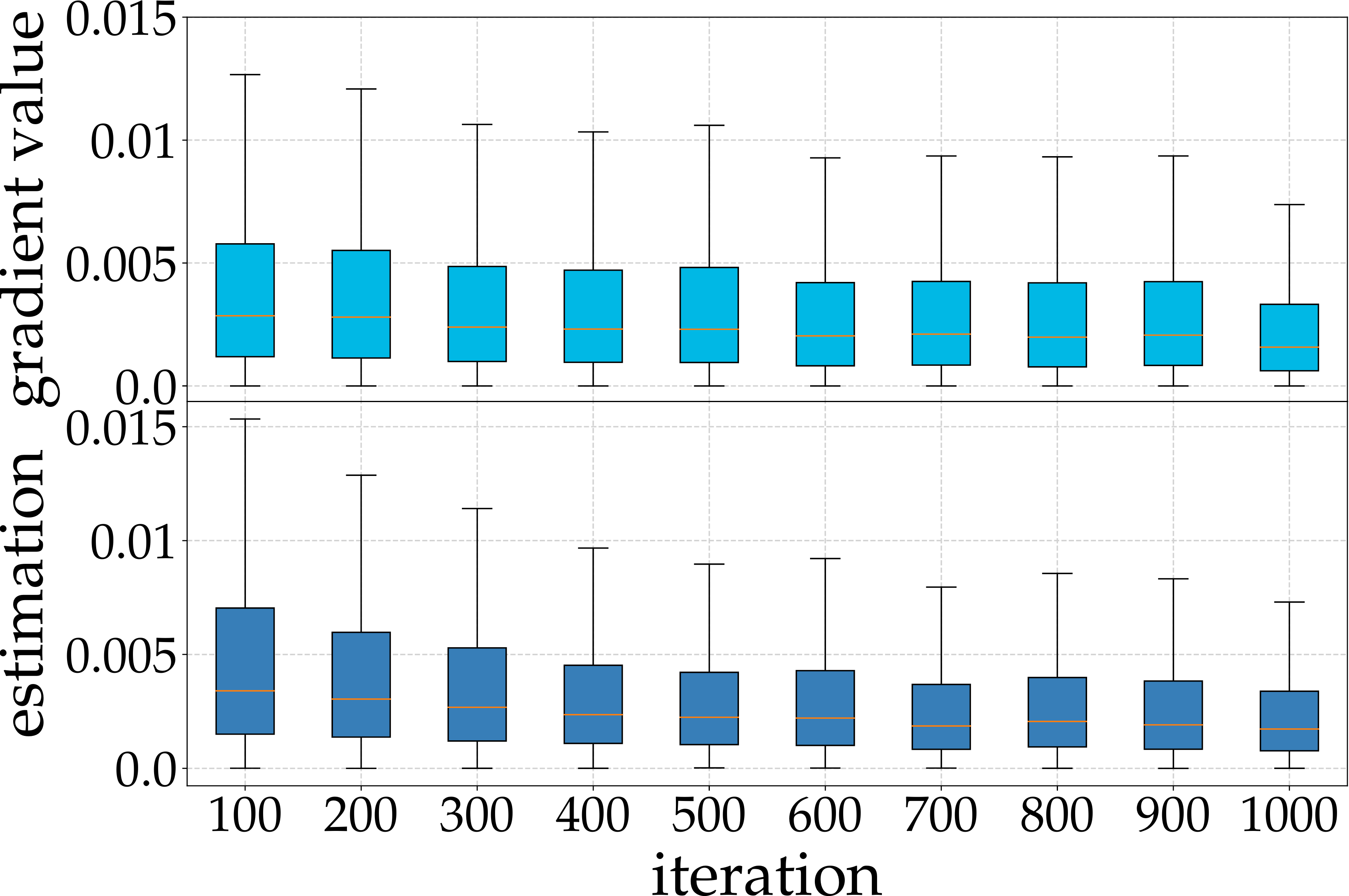}
		\caption{$|g_t|$ and $\sqrt{\expect^\prime[g^2]_{t-1}}$ at different iterations.}
		\label{fig:est_mnist}
	\end{subfigure}
	~
	\begin{subfigure}[b]{0.23\textwidth}
		\includegraphics[width=\textwidth, height=0.115\textheight]{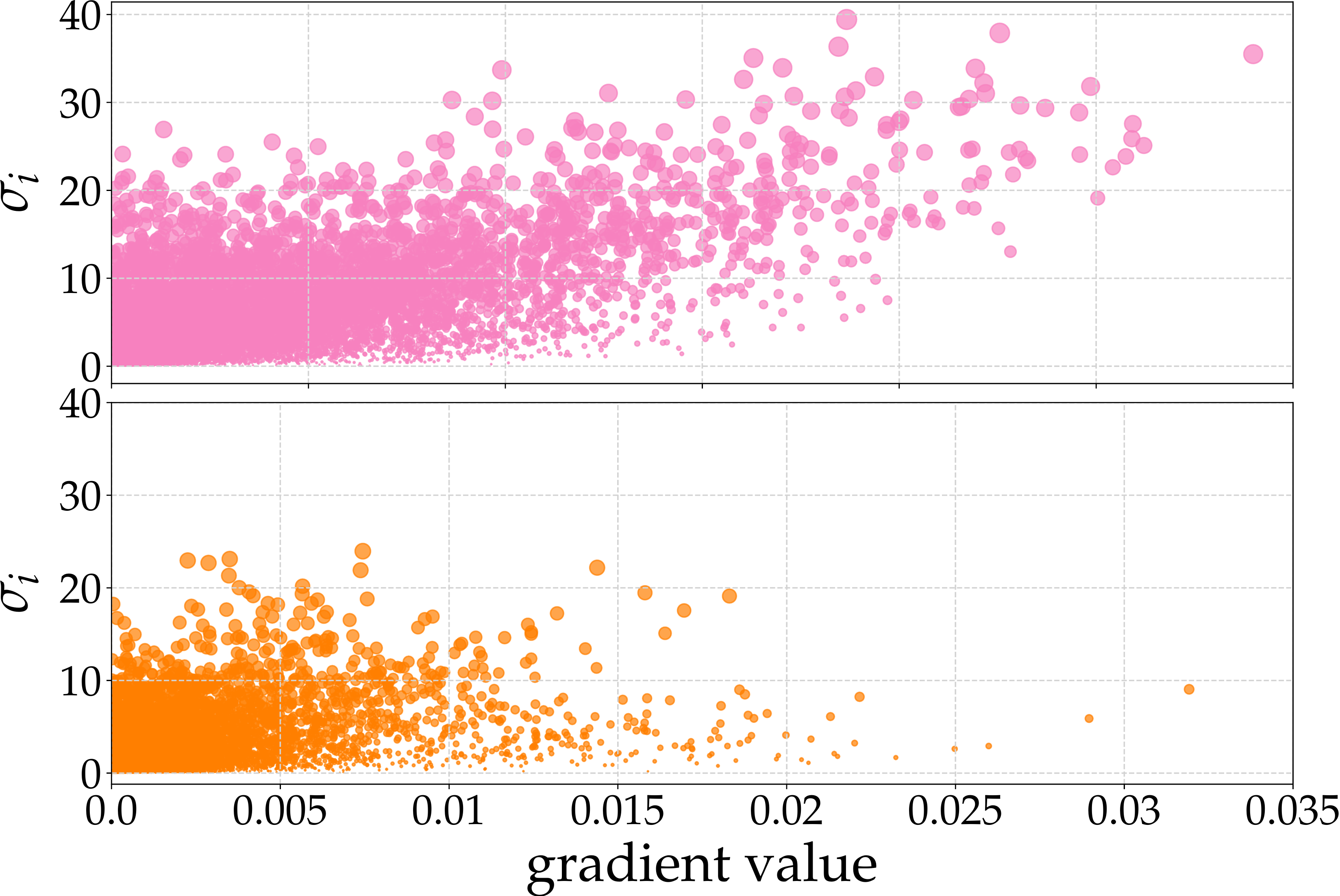}
		\caption{The distribution of $\sigma_i$ at different iterations.}
		\label{fig:sigma_mnist}
	\end{subfigure}
	\caption{The micro benchmarks of key components in \Alg.}
	\label{fig:abl_study}
\end{figure*}

\begin{figure*}[!ht]
	\centering
	\begin{subfigure}[b]{0.23\textwidth}
		\includegraphics[height=0.55\textwidth,width=\textwidth]{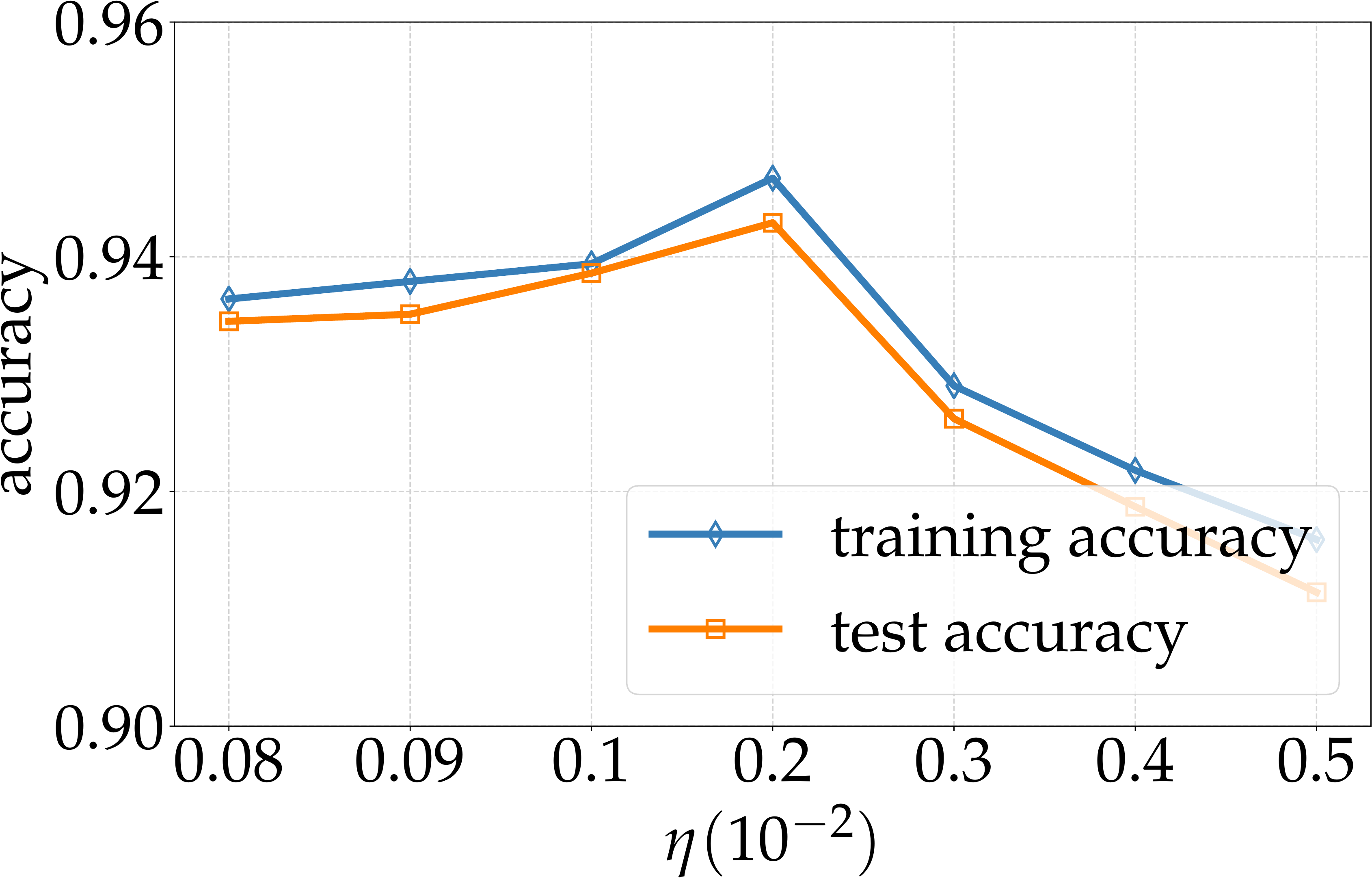}
		\caption{learning rate}\label{fig:mnist_parameters_lr}
	\end{subfigure}
	~
	\begin{subfigure}[b]{0.23\textwidth}
		\includegraphics[height=0.55\textwidth,width=\textwidth]{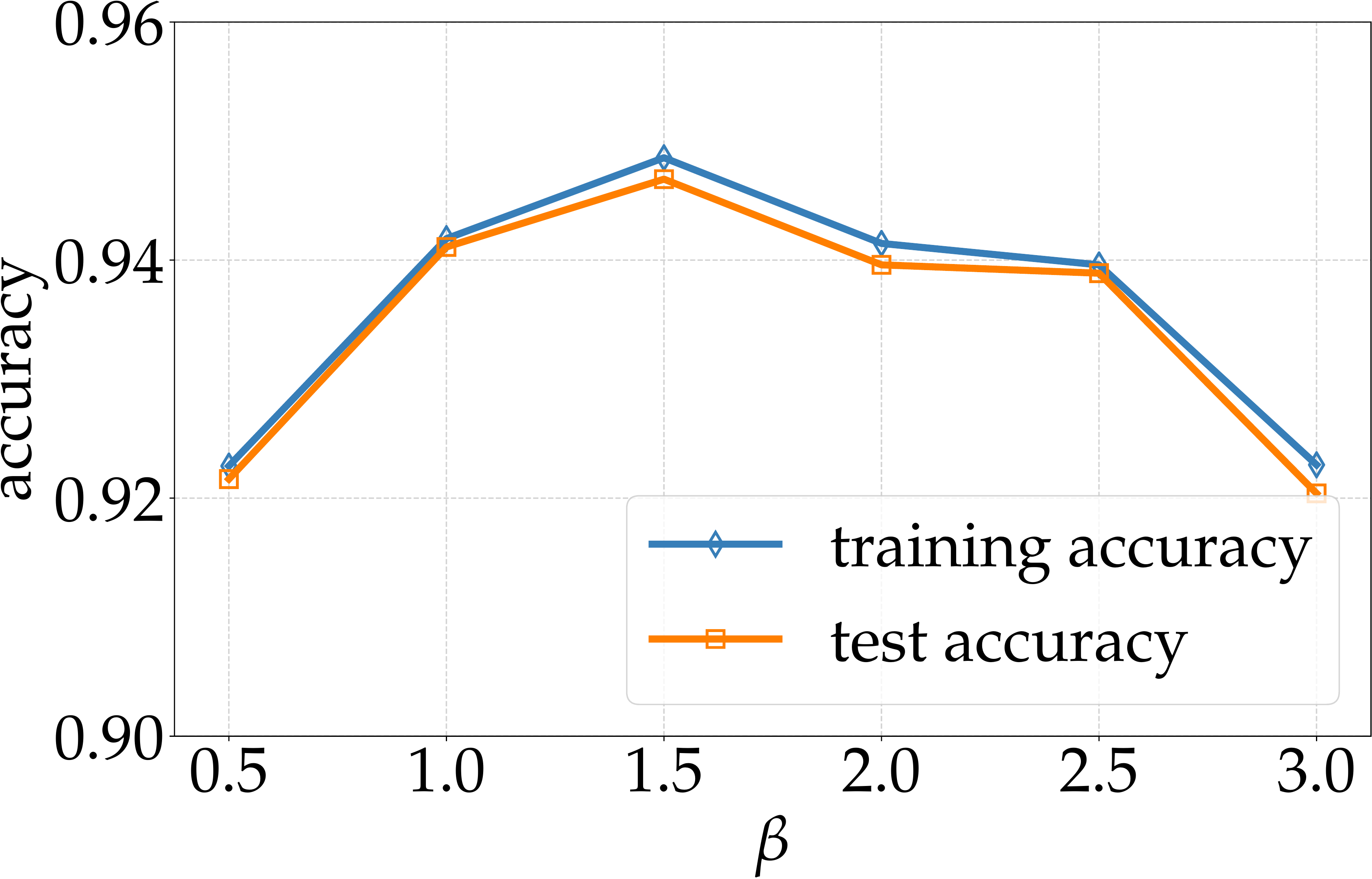}
		\caption{local clipping factor}\label{fig:mnist_parameters_beta}
	\end{subfigure}
	~
	\begin{subfigure}[b]{0.23\textwidth}
		\includegraphics[height=0.555\textwidth,width=\textwidth]{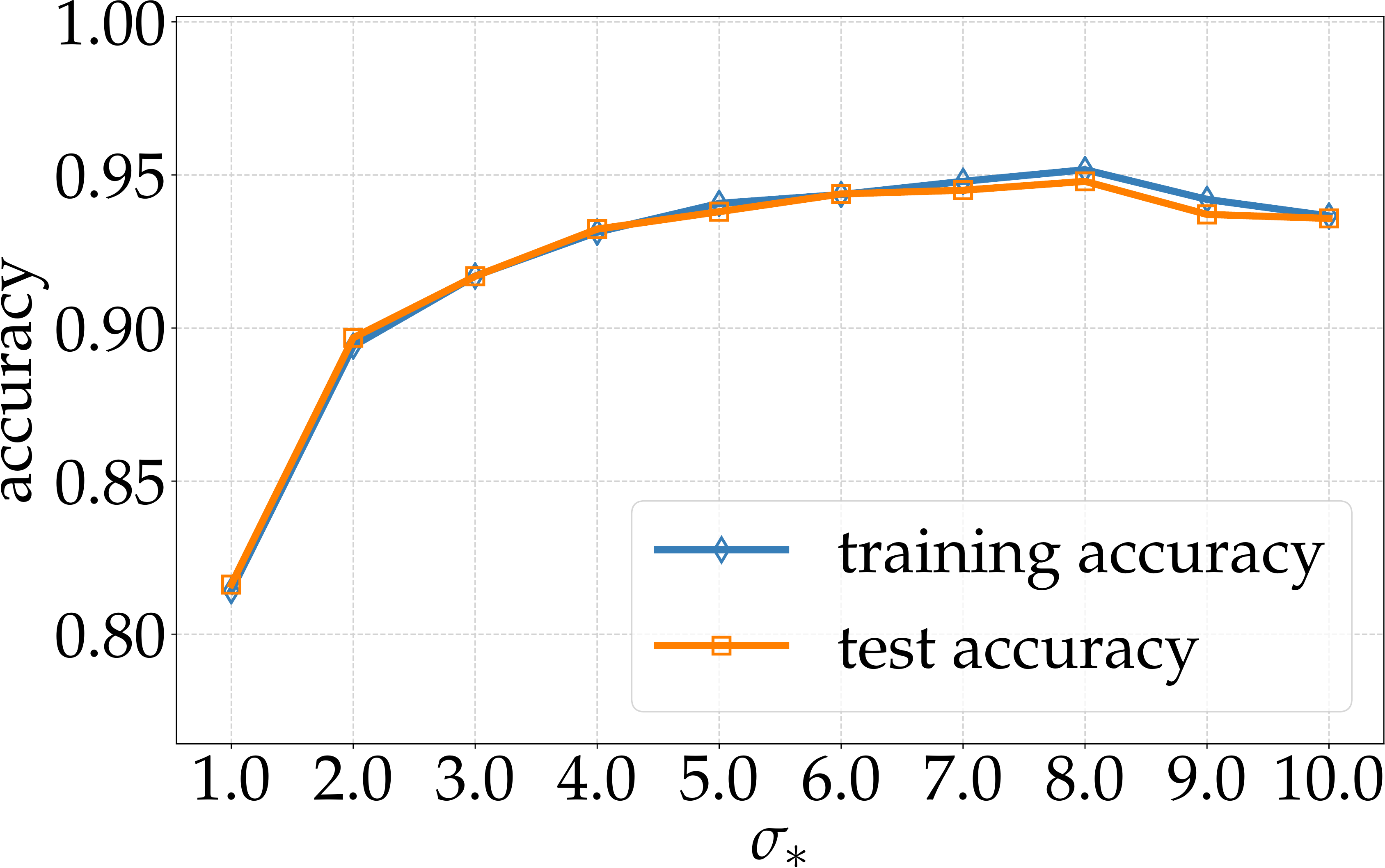}
		\caption{noise level}\label{fig:mnist_parameters_sigma}
	\end{subfigure}
	~
	\begin{subfigure}[b]{0.23\textwidth}
		\includegraphics[height=0.55\textwidth,width=\textwidth]{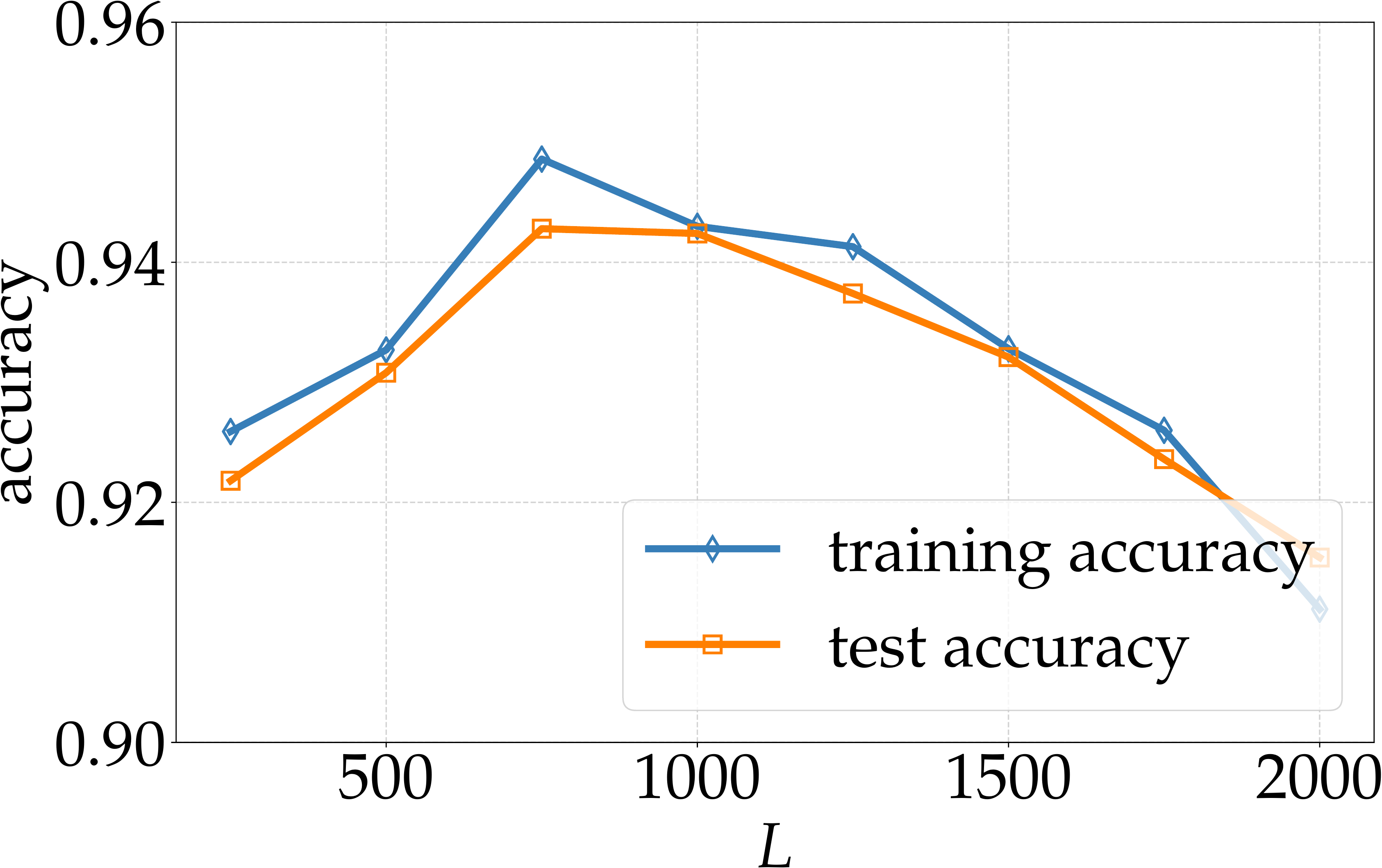}
		\caption{lot size}\label{fig:mnist_parameters_lot}
	\end{subfigure}

	\caption{The sensitivity of different parameters on the accuracy of \Alg.}\label{fig:mnist_sensitivity}
\end{figure*}


Considering there are two components in \Alg, namely the adaptive learning rate and the adaptive noise, we study their independent contribution to the final performance in this subsection. We call the methods with only one component $\textsc{AdaL}$ (only adaptive learning rate) and $\textsc{AdaN}$ (only adaptive noise) respectively. $\textsc{AdaL}$ uses global clipping method to clip the original gradient and samples noise for each dimension from the same Gaussian distribution. The only difference between $\textsc{AdaL}$ and $\textsc{DpSgd}$ is that $\textsc{AdaL}$ adjusts the learning rate based on \cref{eq:rmsprop}. As to $\textsc{AdaN}$, we implement it directly by removing the adaptive learning rate part in \Alg. For the experiment on MNIST, the learning rate is set to 0.1 and 0.001 for $\textsc{AdaN}$ and $\textsc{AdaL}$ respectively. Other settings are the same as \cref{fig:mnist_acc_large_noise}. On CIFAR-10, the learning rate is set to 0.05 and 0.001 for $\textsc{AdaN}$ and $\textsc{AdaL}$ respectively and other settings are the same as \cref{fig:cifar_a}.

\xu{\emph{Both adaptive components contribute to the performance gain of \Alg while the adaptive noise component contributes more.} As illustrated in \cref{fig:ind_mnist_acc_large_noise} and \cref{fig:ind_cifar_acc_large_noise},} we observe that $\textsc{AdaN}$ achieves a higher accuracy than $\textsc{AdaL}$, showing that, compared with the adaptive learning rate, the adaptive noise component has a more significant impact on the performance of \Alg.

Since \Alg depends on whether the estimation of $\sqrt{\expect^\prime[g^2]_{t-1}}$ is accurate for $|g_{t}|$, we further conducted another experiment to verify its effectiveness under the settings as the same as \cref{fig:mnist_acc_large_noise}.

\emph{The estimation given by $\sqrt{\expect^\prime[g^2]_{t-1}}$ is accurate and can reflect the changing trend of  $|g_{t}|$.} The results are shown in \cref{fig:est_mnist}, from which we observe that the distribution of $\sqrt{\expect^\prime[g^2]_{t-1}}$ (the figure below) is close to that of $|g_{t}|$ (the figure above) at each iteration. Also, from the $100$\textsuperscript{th} iteration to the $1000$\textsuperscript{th} iteration, both the gradient value and the estimated value show a declining trend.

Additionally, as aforementioned, given the expected sensitivity of the gradient decreases with the training progresses, each $\sigma_i$ will also decrease gradually. We select the $\sigma_i$ in \Alg at the $100$\textsuperscript{th} iteration and the $1000$\textsuperscript{th} iteration respectively to illustrate this property. All the settings are the same as \cref{fig:mnist_acc_large_noise}.

\xu{\emph{The standard deviation of most noise distributions for each coordinate of the gradient gradually decreases as the training progresses.}} As shown in \cref{fig:sigma_mnist} (the area of each circle is proportional to the value of $\sigma_i$), we observe that most $\sigma_i$ at the $1000$\textsuperscript{th} iteration (the figure below) are more concentrated in the range less than 10 while there are many $\sigma_i$ distributed from 10 to 20 at the $100$\textsuperscript{th} iteration (the figure above). Therefore, the noise distributions in \Alg are adaptive not only to each dimension of the gradient, but also to different iterations during training.

\subsection{Impact of Parameters}
In this set of experiments on the MNIST dataset, we study
 the impact of the learning rate, the local clipping factor, the noise level and the lot size on the accuracy. In all experiments, if not specified, we set the learning rate $\eta=0.002$, lot size $L=600$, local clipping factor $\beta=1.2$, $\ell^2$-norm clipping bound $C=4.0$, noise level $\sigma_*=8.0, \sigma_{p}=16.0$, and privacy level $\epsilon=0.5, \delta=1e^{-5}$.

\textbf{Learning Rate.} As shown in \cref{fig:mnist_parameters_lr}, the accuracy stays consistently above 93\%, irrespective of the learning rate ranging from  $0.08\times10^{-2}$ to $0.5\times10^{-2}$. When the learning rate is lower than  $2\times10^{-2}$, the accuracy increases with the learning rate. When it is higher than  $2\times10^{-2}$, a higher learning rate results in a lower accuracy.

\textbf{Local Clipping Factor.} The local clipping factor $\beta$ controls the scale of each dimension of the gradient. Note that the estimation $\sqrt{\expect[^\prime[g^2]_{t-1}]}$ is not equal to $|g_t|$, a smaller $\beta$ drops more information of the gradient since some dimensions will be clipped as ${s_i}$ or $-{s_i}$. \cref{fig:mnist_parameters_beta} shows that the performance of \Alg is relatively resistant to the local clipping factor and a reasonable setting for $\beta$ is to take a value slightly greater than one. However, due to the adaptive noise scheme of \Alg, a large value of $\beta$ will not raise the noise level drastically (especially because when the gradient is close to 0 as the training converges, $s_i$ will also be close to 0 whatever $\beta$ takes). 

\textbf{Noise Scale.} The noise scale determines the amount of Gaussian noise added to the update term at each step. Although a smaller noise scale mitigates the effect of noises, it results in fewer training steps and the model is hard to converge. Meanwhile, setting a larger noise scale allows more training steps, but excessive noise will ruin the original gradient. Results achieved with different noise scales are shown in \cref{fig:mnist_parameters_sigma}, where our model attains
the most superior performance with $\sigma_*=7.0$. Compared with \AlgDpSgd,
which reaches the highest accuracy at $\sigma_*=4.0$ \cite{abadi2016deep}, we
attribute the much larger value of $\sigma_*$ of \Alg
to adaptive noises which alleviate the impact of the differential privacy mechanism and a larger noise scale yields more training steps.

\textbf{Lot Size.} The lot size controls the sampling ratio. A large lot size yields a higher sampling ratio and reduces the number of training steps. If the noise intensity is fixed, a smaller lot size results in a more notable effect of
 Gaussian noise. \cref{fig:mnist_parameters_lot} shows the accuracy vs.\
the lot size. It can be observed that as the lot size grows, the performance increases initially, peaks at $L=800$, and finally declines. The result agrees with the above analysis that too high too low sampling ratio will incur performance degradation.



\section{Conclusions}
\label{sec:Conclusion}
In this paper, our key contribution is three-fold. First, we propose a differentially private deep learning algorithm which leads to a faster convergence and higher accuracy in comparison with prior methods. Second, we intuitively analyze advantages of \Alg over \AlgDpSgd and mathematically prove that \Alg satisfies differential privacy by more advanced analytic tools. Third, we applied \Alg, \AlgDpSgd, \textsc{AGD} and \textsc{DpOpt} for model training of deep learning networks with real-world datasets and experimentally evaluate the better performance of \Alg.



\begin{thebibliography}{10}
\providecommand{\url}[1]{#1}
\csname url@samestyle\endcsname
\providecommand{\newblock}{\relax}
\providecommand{\bibinfo}[2]{#2}
\providecommand{\BIBentrySTDinterwordspacing}{\spaceskip=0pt\relax}
\providecommand{\BIBentryALTinterwordstretchfactor}{4}
\providecommand{\BIBentryALTinterwordspacing}{\spaceskip=\fontdimen2\font plus
\BIBentryALTinterwordstretchfactor\fontdimen3\font minus
  \fontdimen4\font\relax}
\providecommand{\BIBforeignlanguage}[2]{{%
\expandafter\ifx\csname l@#1\endcsname\relax
\typeout{** WARNING: IEEEtran.bst: No hyphenation pattern has been}%
\typeout{** loaded for the language `#1'. Using the pattern for}%
\typeout{** the default language instead.}%
\else
\language=\csname l@#1\endcsname
\fi
#2}}
\providecommand{\BIBdecl}{\relax}
\BIBdecl

\bibitem{wang2016csi}
X.~Wang, L.~Gao, and S.~Mao, ``Csi phase fingerprinting for indoor localization
  with a deep learning approach,'' \emph{IEEE Internet of Things Journal},
  vol.~3, no.~6, pp. 1113--1123, 2016.

\bibitem{wang2017spatiotemporal}
J.~Wang, J.~Tang, Z.~Xu, Y.~Wang, G.~Xue, X.~Zhang, and D.~Yang,
  ``Spatiotemporal modeling and prediction in cellular networks: A big data
  enabled deep learning approach,'' in \emph{Proceddings of 36th Annual IEEE
  International Conference on Computer Communications (INFOCOM)}, 2017, pp.
  1--9.

\bibitem{zhou2018deep}
Y.~Zhou, M.~Han, L.~Liu, J.~S. He, and Y.~Wang, ``Deep learning approach for
  cyberattack detection,'' in \emph{Proceddings of 37th IEEE International
  Workshop on Computer Communications (INFOCOM)}, 2018, pp. 262--267.

\bibitem{fredrikson2014privacy}
M.~Fredrikson, E.~Lantz, S.~Jha, S.~Lin, D.~Page, and T.~Ristenpart, ``Privacy
  in pharmacogenetics: An end-to-end case study of personalized warfarin
  dosing,'' in \emph{Proceedings of 23rd USENIX Security Symposium (USENIX
  Security)}, 2014, pp. 17--32.

\bibitem{dwork2006calibrating}
C.~Dwork, F.~McSherry, K.~Nissim, and A.~Smith, ``Calibrating noise to
  sensitivity in private data analysis,'' in \emph{Theory of cryptography
  conference}.\hskip 1em plus 0.5em minus 0.4em\relax Springer, 2006, pp.
  265--284.

\bibitem{bost2015machine}
R.~Bost, R.~A. Popa, S.~Tu, and S.~Goldwasser, ``Machine learning
  classification over encrypted data.'' in \emph{Proceedings of 22nd Annual
  Network \& Distributed System Security Symposium (NDSS)}, vol. 4324, 2015, p.
  4325.

\bibitem{zhang2018privacy}
T.~Zhang, ``Privacy-preserving machine learning through data obfuscation,''
  \emph{arXiv preprint arXiv:1807.01860}, 2018.

\bibitem{kairouz2015secure}
P.~Kairouz, S.~Oh, and P.~Viswanath, ``Secure multi-party differential
  privacy,'' in \emph{29th Conference on Neural Information Processing Systems
  (NeurIPS)}, 2015, pp. 1999--2007.

\bibitem{shokri2015privacy}
R.~Shokri and V.~Shmatikov, ``Privacy-preserving deep learning,'' in
  \emph{Proceedings of 22nd ACM Conference on Computer and Communications
  Security (CCS)}, 2015, pp. 1310--1321.

\bibitem{abadi2016deep}
M.~Abadi, A.~Chu, I.~Goodfellow, H.~B. McMahan, I.~Mironov, K.~Talwar, and
  L.~Zhang, ``Deep learning with differential privacy,'' in \emph{Proceedings
  of 23rd ACM Conference on Computer and Communications Security (CCS)}, 2016,
  pp. 308--318.

\bibitem{lee2018concentrated}
J.~Lee and D.~Kifer, ``Concentrated differentially private gradient descent
  with adaptive per-iteration privacy budget,'' in \emph{Proceedings of 24th
  ACM SIGKDD Conference on Knowledge Discovery and Data Mining (KDD)}, 2018,
  pp. 1656--1665.

\bibitem{koskela2018learning}
A.~Koskela and A.~Honkela, ``Learning rate adaptation for differentially
  private stochastic gradient descent,'' \emph{arXiv preprint
  arXiv:1809.03832}, 2018.

\bibitem{xiang2019differentially}
L.~Xiang, J.~Yang, and B.~Li, ``Differentially-private deep learning from an
  optimization perspective,'' in \emph{Proceedings of 38th Annual IEEE
  Conference on Computer Communications (INFOCOM)}, 2019, pp. 559--567.

\bibitem{mironov2017renyi}
I.~Mironov, ``R\'{e}nyi differential privacy,'' in \emph{Proceedings of 30th
  IEEE Computer Security Foundations Symposium (CSF)}, 2017, pp. 263--275.

\bibitem{wang2018subsampled}
Y.-X. Wang, B.~Balle, and S.~Kasiviswanathan, ``Subsampled r\'enyi differential
  privacy and analytical moments accountant,'' \emph{arXiv preprint
  arXiv:1808.00087}, 2018.

\bibitem{lecun1998gradient}
Y.~LeCun, L.~Bottou, Y.~Bengio, and P.~Haffner, ``Gradient-based learning
  applied to document recognition,'' \emph{Proceedings of the IEEE}, vol.~86,
  no.~11, pp. 2278--2324, 1998.

\bibitem{krizhevsky2009learning}
A.~Krizhevsky and G.~Hinton, ``Learning multiple layers of features from tiny
  images,'' Citeseer, Tech. Rep., 2009.

\bibitem{zhang2018crowdbuy}
L.~Zhang, Y.~Li, X.~Xiao, X.-Y. Li, J.~Wang, A.~Zhou, and Q.~Li, ``Crowdbuy:
  Privacy-friendly image dataset purchasing via crowdsourcing,'' in
  \emph{Proceedings of 37th IEEE Annual Conference on Computer Communications
  (INFOCOM)}, 2018, pp. 2735--2743.

\bibitem{jin2019if}
W.~Jin, M.~Xiao, M.~Li, and L.~Guo, ``If you do not care about it, sell it:
  Trading location privacy in mobile crowd sensing,'' in \emph{Proceedings of
  38th IEEE Annual Conference on Computer Communications (INFOCOM)}, 2019, pp.
  1045--1053.

\bibitem{niu2019making}
C.~Niu, Z.~Zheng, S.~Tang, X.~Gao, and F.~Wu, ``Making big money from small
  sensors: Trading time-series data under pufferfish privacy,'' in
  \emph{Proceedings of 38th IEEE Annual Conference on Computer Communications
  (CCS)}.\hskip 1em plus 0.5em minus 0.4em\relax IEEE, 2019, pp. 568--576.

\bibitem{li2019pedss}
H.~Li, Y.~Yang, Y.~Dou, J.-M.~J. Park, and K.~Ren, ``Pedss: Privacy enhanced
  and database-driven dynamic spectrum sharing,'' in \emph{Proceedings of 38th
  IEEE Annual Conference on Computer Communications (INFOCOM)}.\hskip 1em plus
  0.5em minus 0.4em\relax IEEE, 2019, pp. 1477--1485.

\bibitem{phan2016differential}
N.~Phan, Y.~Wang, X.~Wu, and D.~Dou, ``Differential privacy preservation for
  deep auto-encoders: {An} application of human behavior prediction,'' in
  \emph{Proceedings of 30th AAAI Conference on Artifacial Intelligence (AAAI)},
  vol.~16, 2016, pp. 1309--1316.

\bibitem{papadimitriou2017dstress}
A.~Papadimitriou, A.~Narayan, and A.~Haeberlen, ``{DStress}: {E}fficient
  differentially private computations on distributed data.'' in
  \emph{Proceedings of the 12th European Conference on Computer Systems
  (EuroSys)}, 2017, pp. 560--574.

\bibitem{papernot2016semi}
N.~Papernot, M.~Abadi, {\'U}.~Erlingsson, I.~Goodfellow, and K.~Talwar,
  ``Semi-supervised knowledge transfer for deep learning from private training
  data,'' in \emph{Proceedings of 5th International Conference on Learning
  Representations (ICLR)}, 2017.

\bibitem{phan2017adaptive}
N.~Phan, X.~Wu, H.~Hu, and D.~Dou, ``Adaptive laplace mechanism: differential
  privacy preservation in deep learning,'' in \emph{Proceedings of the IEEE
  International Conference on Data Mining (ICDM)}, 2017, pp. 385--394.

\bibitem{papernot2018scalable}
N.~Papernot, S.~Song, I.~Mironov, A.~Raghunathan, K.~Talwar, and
  {\'U}.~Erlingsson, ``Scalable private learning with {PATE},'' in
  \emph{Proceedings of 6th International Conference on Learning Representations
  (ICLR)}, 2018.

\bibitem{collet2018boosting}
S.~Collet, R.~Dadashi, Z.~N. Karam, C.~Liu, P.~Sobhani, Y.~Vahlis, and J.~C.
  Zhang, ``Boosting model performance through differentially private model
  aggregation,'' \emph{arXiv preprint arXiv:1811.04911}, 2018.

\bibitem{wang2018geographic}
L.~Wang, G.~Qin, D.~Yang, X.~Han, and X.~Ma, ``Geographic differential privacy
  for mobile crowd coverage maximization,'' in \emph{Proceedings of 32nd AAAI
  Conference on Artificial Intelligence (AAAI)}, 2018.

\bibitem{mao2018privacy}
Y.~Mao, S.~Yi, Q.~Li, J.~Feng, F.~Xu, and S.~Zhong, ``A privacy-preserving deep
  learning approach for face recognition with edge computing,'' in \emph{USENIX
  Workshop on Hot Topics in Edge Computing (HotEdge)}, 2018.

\bibitem{balle2018improving}
B.~Balle and Y.-X. Wang, ``Improving the {Gaussian} mechanism for differential
  privacy: Analytical calibration and optimal denoising,'' \emph{arXiv preprint
  arXiv:1805.06530}, 2018.

\bibitem{wiki:Test_functions_for_optimization}
``{Test functions for optimization} --- {W}ikipedia{,} the free encyclopedia,''
  \url{http://en.wikipedia.org/w/index.php?title=Test_functions_for_optimization}.

\bibitem{goodfellow2015efficient}
I.~Goodfellow, ``Efficient per-example gradient computations,'' \emph{arXiv
  preprint arXiv:1510.01799}, 2015.

\bibitem{dwork2014algorithmic}
C.~Dwork, A.~Roth \emph{et~al.}, ``The algorithmic foundations of differential
  privacy,'' \emph{Foundations and Trends{\textregistered} in Theoretical
  Computer Science}, vol.~9, no. 3--4, pp. 211--407, 2014.

\bibitem{abadi2016tensorflow}
M.~Abadi, P.~Barham, J.~Chen, Z.~Chen, A.~Davis, J.~Dean, M.~Devin,
  S.~Ghemawat, G.~Irving, M.~Isard \emph{et~al.}, ``Tensorflow: {A} system for
  large-scale machine learning,'' in \emph{Proceedings of 12th USENIX Symposium
  on Operating Systems Design and Implementation (OSDI)}, vol.~16, 2016, pp.
  265--283.

\bibitem{dwork2014analyze}
C.~Dwork, K.~Talwar, A.~Thakurta, and L.~Zhang, ``Analyze gauss: optimal bounds
  for privacy-preserving principal component analysis,'' in \emph{Proceedings
  of 46th Annual ACM Symposium on Theory of Computing (STOC)}, 2014, pp.
  11--20.

\bibitem{nair2010rectified}
V.~Nair and G.~E. Hinton, ``Rectified linear units improve restricted boltzmann
  machines,'' in \emph{Proceedings of 27th International Conference on Machine
  Learning (ICML)}, 2010, pp. 807--814.

\bibitem{krizhevsky2012imagenet}
A.~Krizhevsky, I.~Sutskever, and G.~E. Hinton, ``Imagenet classification with
  deep convolutional neural networks,'' in \emph{Proceedings of 26th Conference
  on Neural Information Processing Systems (NeurIPS)}, 2012, pp. 1097--1105.

\bibitem{jarrett2009best}
K.~Jarrett, K.~Kavukcuoglu, Y.~LeCun \emph{et~al.}, ``What is the best
  multi-stage architecture for object recognition?'' in \emph{Proceedings of
  IEEE 12th International Conference on Computer Vision (ICCV)}, 2009, pp.
  2146--2153.

\end{thebibliography}

\end{document}